\documentclass{article}
\usepackage[utf8]{inputenc}
\usepackage{amssymb}
\usepackage{amsmath, amsfonts}
\usepackage{forest}
\usepackage{amsthm}
\usepackage{mathtools}

\usepackage{enumitem}
\usepackage{caption}

\usepackage{url}
\usepackage{verbatim}
\usepackage{physics}

\usepackage{tikz}
\usetikzlibrary{shapes.geometric, arrows}

\usepackage{graphicx}
\usepackage{framed}
\usepackage{multicol}
\usepackage{multirow}
\usepackage{listings}
\usepackage[paper=a4paper,margin=1in]{geometry}
\usepackage[inline]{asymptote}

\newcommand{\pavel}[1]{}
\newcommand{\sreejato}[1]{}
\newcommand{\arkadev}[1]{}

\newtheorem{theorem}{Theorem}
\newtheorem{lemma}{Lemma}
\newtheorem{corollary}{Corollary}
\newtheorem{definition}{Definition}
\newtheorem{observation}{Observation}
\newtheorem{claim}{Claim}

\newtheorem{problem}{Problem}

\newtheorem*{lemma*}{Lemma}
\newtheorem*{theorem*}{Theorem}

\newcommand{\N}{\mathbb{N}}
\newcommand{\TV}{d_{\text{TV}}}

\newcommand{\F}{\mathbb{F}}
\newcommand{\codim}{\text{co-dim}}
\newcommand{\cube}{\{0,1\}}
\newcommand{\Cl}{\text{Cl}}
\newcommand{\VarCl}{\text{VarCl}}
\newcommand{\Space}{\mathcal{S}}
\newcommand{\Span}{\text{Span}}
\newcommand{\Row}{\mathcal{R}}

\newcommand{\SP}{\mathsf{SP}}
\newcommand{\Stone}{\text{Stone}}
\newcommand{\V}{\mathcal{V}}
\newcommand{\Supp}{\text{Supp}}

\newcommand{\Search}{\textit{Search}}
\newcommand{\IP}{\mathsf{IP}}

\newcommand{\Post}{\textit{Post}}
\newcommand{\Pre}{\textit{Pre}}
\newcommand{\Prob}{\mathsf{P}}
\newcommand{\cP}{\mathcal{P}}
\newcommand{\cT}{\mathcal{T}}

\tikzstyle{redcircle}= [circle, draw=black, fill=red!30]
\tikzstyle{bluecircle}= [circle, draw=black, fill=blue!30]
\tikzstyle{blackcircle}= [circle, draw=black, fill=black!30]

\title{Exponential Separation Between Powers of Regular and General Resolution Over Parities}
\author{Sreejata Kishor Bhattacharya\thanks{TIFR-Mumbai, email: sreejata.bhattacharya@tifr.res.in} \and Arkadev Chattopadhyay\thanks{TIFR - Mumbai, email: arkadev.c\@@tifr.res.in. Partially funded by the Department of Atomic Energy and a Google India Research Award.} \and Pavel Dvořák\thanks{TIFR-Mumbai \& Charles University, Prague, email: koblich\@@iuuk.mff.cuni.cz. Work done entirely at TIFR. Supported by Czech Science Foundation GAČR grant \#22-14872O.}}

\date{ }

\newcommand{ \prob} {\text{Pr}}
\newtheorem{remark}{Remark}[section]

\theoremstyle{definition}
\theoremstyle{remark}
\numberwithin{equation}{section}

\tikzstyle{startstop} = [rectangle, rounded corners, 
minimum width=3cm, 
minimum height=1cm,
text centered, 
draw=black, 
fill=white!30]

\tikzstyle{startstopdashed} = [rectangle, dashed, rounded corners, 
minimum width=3cm, 
minimum height=1cm,
text centered, 
draw=black, 
fill=white!30]

\tikzstyle{process} = [rectangle, 
minimum width=8cm, 
minimum height=1cm, 
text centered, 
text width=7cm, 
draw=black, 
fill=white!30]

\tikzstyle{arrow} = [thick,->,>=stealth]

\tikzstyle{abcd}= [rectangle]

\begin{document}
\maketitle

\begin{abstract}
    Proving super-polynomial lower bounds on the size of proofs of unsatisfiability of Boolean formulas using resolution over parities is an outstanding problem that has received a lot of attention after its introduction by Raz and Tzamaret \cite{RT08}. Very recently, Efremenko, Garl{\'i}k and Itsykson \cite{EGI23} proved the first exponential lower bounds on the size of ResLin proofs that were additionally restricted to be \emph{bottom-regular}. We show that there are formulas for which such regular ResLin proofs of unsatisfiability continue to have exponential size even though there exist short proofs of their unsatisfiability in ordinary, non-regular resolution. This is the first super-polynomial separation between the power of general ResLin and and that of regular ResLin for any natural notion of regularity. 
    
    Our argument, while building upon the work of Efremenko et al., uses additional ideas from the literature on \emph{lifting theorems}.
\end{abstract}

\section{Introduction}
One of the most basic and well studied proof systems in propositional proof complexity is resolution. Its weakness by now is reasonably well understood after years of research. Yet, this understanding is quite fragile as natural and simple strengthenings of resolution quickly pose challenges that remain outstanding. One such system is resolution over parities, introduced by Raz and Tzameret \cite{RT08}, which has been abbreviated as ResLin. More precisely, a \emph{linear} clause is a disjunction of affine equations, generalizing the notion of ordinary clauses. If $A$ and $B$ are two such linear disjunctions and $\ell$ is a linear form, then the inference rule of ResLin derives the linear clause $A \vee B$ from clauses $A \vee (\ell = 1)$ and $B \vee (\ell=0)$. To appreciate the power of this system, let us recall that a linear clause, unlike an ordinary clause, can be expressed using many different bases. Indeed, no super-polynomial lower bounds on the size of general proofs in this system is currently known for any explicit unsatisfiable Boolean formula.

Progress was first made in the work of Itsykson and Sokolov \cite{IS14} when they proved exponential lower bounds on the size of tree-like ResLin proofs for central tautologies including  the Pigeonhole Principle and Tseitin formulas over expanding graphs. Proving such lower bounds was systematized, only recently, in the independent works of Chattopadhyay, Mande, Sanyal and Sherif \cite{CMSS23} and that of Beame and Koroth \cite{BK23}. These works could be used to conclude that tree-like ResLin proofs are exponentially weaker than general ResLin proofs.

In the world of ordinary resolution, it is known that there is an intermediate proof system, known as \emph{regular} resolution, whose power strictly lies in between tree-like and general proofs. In the graph of a regular resolution proof, no derivation path from an axiom clause to the final empty clause resolves a variable of the formula more than once. In the dual view of searching for a falsified clause, this corresponds precisely to read-once branching programs, where no source to sink path queries a variable more than once. Taking cue from this,  
Gryaznov, Pudl{\'a}k and Talebanfard \cite{GPT22} introduced models of read-once linear branching programs (ROLBP), to capture notions of regularity in ResLin. They identified two such notions that extend the notion of regularity in ordinary resolution, or the read-once property of branching programs. Consider a node $v$ of an ROLBP. Let $\Pre(v)$ denote the vector space spanned by all the linear queries that appear in some path from the source node to $v$. Similarly, let $\Post(v)$ denote the space spanned by all linear queries that lie in some path from $v$ to a sink node. In the most restrictive notion, called strongly regular proofs or strongly read-once linear branching programs, $\Pre(v)$ and $\Post(v)$ have no non-trivial intersection for every $v$. In a more relaxed notion, called weak regularity or weakly ROLBP, the linear query made at node $v$ is not contained in $\Pre(v)$. Gryaznov et al. \cite{GPT22} were able to prove an exponential lower bound on the size of strongly ROLBP for computing a \emph{function}, using the notion of directional affine dispersers. However, their argument is not known to work for search problems for even strongly ROLBP. 

There is another natural notion of weakly read-once linear branching programs (dually, weakly regular ResLin  ) that complements the notion defined in \cite{GPT22}. In this notion, we forbid the linear query made at a node $v$ of the branching program to lie in an affine space $\Post(u)$, for each $u$ that is a child of $v$. We will call this notion bottom-read-once (bottom-regular proofs) and  the former notion of Gryaznov et al as top-read-once (top-regular proofs). Both are generalizations of strongly read-once linear branching programs (strongly regular ResLin proofs). 
 Very recently, Efremenko, Garlik and Itsykson \cite{EGI23} proved the first exponential lower bounds on the size of bottom-regular ResLin proofs. The tautology they used was the Binary Pigeonhole Principle (BPHP). It is plausible that the BPHP remains hard for general ResLin proofs. One way to prove such a bound would be to show that every general ResLin proof could be converted to a bottom-regular ResLin proof with a (quasi-)polynomial blow up. 

Our main result is a strong refutation of that possibility. We show that bottom-regular ResLin proofs require exponential size blow-up to simulate non-regular proofs even in the ordinary resolution system. The formulas we use are twists of certain formulas used by Alekhnovich, Johannsen, Pitassi and Urquhart \cite{AJPU02}  for providing separation between regular and general resolution.  Alekhnovich et al. provided two formulas for proving such a separation. In the second one, the starting point is a pebbling formula on pyramid graphs. It turns out that they are easy for regular resolution. To get around that, they consider \emph{stone formula} for pyramid graphs, that they prove is hard for regular resolution while remaining easy for general resolution. We further obfuscate such stone formulas using another idea of Alekhnovich et al. that appears in the construction of their  first formula.  Finally, we \emph{lift} these formulas by logarithmic size Inner-product ($\IP$) gadgets.

This lifted formula presents fresh difficulties to be overcome to carry out the implicit technique of Efremenko et al. We overcome them by exploitng two properties of the Inner-product. First, we exploit the property that $\IP$ has low discrepancy, invoking a result of Chattopadhyay, Filmus, Koroth, Meir and Pitassi \cite{CFKMP21}. Second, we use the fact that $\IP$ has the stifling property, inspired by the recent work of Chattopadhyay, Mande, Sanyal and Sherif \cite{CMSS23}.

Our unsatisfiable formula that yields a separation of resolution and bottom-regular ResLin is a stone formula for a pyramid graph $G_n$ of $n$ levels lifted by an inner-product gadget $\IP: \cube^b \to \cube$.
We denote this formula as $\SP_n \circ \IP$ and it is defined over $M := 2N^2 \cdot b$ variables, where $b = \Theta(\log n)$ and $N = n(n+1) / 2$ is the number of vertices of the pyramid graph $G_n$.
For exact definition of $\SP_n \circ \IP$, see Section~\ref{sec:HardProblem}.

\begin{theorem}
    \label{thm:MainResult}
    The formula $\SP_n \circ \IP$ admits a resolution refutation of length that is polynomial in $M$ but any bottom-regular ResLin refutation of it must have length at least $2^{\Omega(M^{1/12} / \log^{1 + \varepsilon} M)}$ for $\varepsilon = 1/12$.
\end{theorem}

\paragraph{Comparision with Efremenko et al.~\cite{EGI23}:} Our work builds upon the very recent work of Efremenko, Garl{\'i}k and Itsykson, who proved an exponential lower bound for the regular linear resolution complexity of the formula $\text{Binary-PHP}_{n}^{n+1}$. A crucial property of this formula that they use is that if we sample an assignment to the variables from the uniform distribution, with high probability one needs to make at least $n^{\Omega(1)}$ bit-queries to locate a falsified clause. Later, they use the following simple property of the uniform distribution: let $A$ be an affine subspace of co-dimension $r$. Then, the probability mass of $A$ under uniform distribution is very small (inverse-exponential in $r$). Call this property (*). 

Our goal is to show an exponential lower bound on the regular linear resolution complexity of a formula that has a small resolution refutation. A candidate formula would be a CNF which exhibits exponential separation between resolution and regular resolution. Some such formulas are $\text{MGT}_{n,\rho}$ and \textit{stone formulas} with auxiliary variables to keep width of clauses short (both defined in Alekhnovich et al. \cite{AJPU02}). However, all such formulas have constant width -- and therefore, a uniformly random assignment falsifies a constant fraction of clauses. It follows that for both these formulas there is a query algorithm making only constantly many queries which finds a falsified clause with high probability under the uniform distribution. Thus, directly adapting the argument of Efremenko et al. would not work for these formulas. 

Our main observation is that property (*) continues to hold for a much larger class of distributions than the uniform distribution when the base formula is \textit{lifted} with an appropriate gadget. More precisely, if we take any distribution $\mu$ on the assignments of the base formula and let $\mu'$ be its \textit{uniform lift}, property (*) holds for $\mu'$. We are now free to choose \textit{any} distribution on the assignments of the base formula for which locating a falsifying axiom requires many queries on average (this is just a sketch; we actually need something slightly stronger). This gives us enough freedom to construct appropriate distributions. Some more ingredients are required to make this idea work; we explain them in the subsequent sections.

\paragraph{Some Other Related Work:} Following up on the work by Alekhnovich et al \cite{AJPU02}, Urquhart \cite{U11} proved a stronger separation between the length of regular and general resolution proofs.  Much more recently, Vinyals, Elffer, Johannsen and Nordstr{\"o}m \cite{VEJN20}, designed a different formula for showing an even stronger separation between regular and general resolutions. The constructions of Urquhart \cite{U11} as well as that of Vinyals et al \cite{VEJN20} are somewhat related to the hard formulas that we construct in this paper. We talk about them more at the end of Section~\ref{sec:HardProblem} after describing our construction in detail. In another direction, the model of read-once linear branching programs, introduced by Gryaznov, Pudl{\'a}k and Talebanfard \cite{GPT22}, spawned research in directional affine extractors and pseudo-randomness first by Chattopadhyay and Liao \cite{CL23}, and then further by Li and Zhong \cite{LZ23}. These work on extractors, while independently interesting, are not known to have consequences for ResLin.



\paragraph{Organization of the Paper:} 
We present our hard formula in the next section and briefly compare it with constructions done in earlier work. 
In Section~\ref{sec:BP_Res}, we define ResLin refutation system and its regular and tree-like restrictions. Further, we present the connection between resolution proof systems and branching programs.
In Section~\ref{sec:linear}, we present some results from linear algebra which we use in our proofs.
In Section~\ref{sec:Outline}, we sketch the outline and main ideas of the proof of our main result, Theorem~\ref{thm:MainResult}.
Then in Section~\ref{sec:UpperBound}, we prove the upper bound part of Theorem~\ref{thm:MainResult}, i.e., our hard formula has short resolution refutation.
We finish our proof in Section~\ref{sec:LowerBound} where we prove the lower bound part of Theorem~\ref{thm:MainResult}, i.e., any bottom-regular ResLin refutation of our hard formula must have exponential length. Finally, in Section~\ref{sec:conclusion}, we conclude with some of the many open problems that our work raises.

\section{ A Formula Hard For Just Regular ResLin}
\label{sec:HardProblem}
Let us first recall the stone formula that was used by Alekhnovich et al.~\cite{AJPU02} for separating the powers of regular and general resolution. The formula that we shall use is a \emph{lift} of this formula by an appropriate gadget. 
Let $G= (V,E)$ be a directed acyclic graph such that it has exactly one root (vertex with indegree 0), $r$, and  every vertex of $G$ has outdegree either 0 or 2.
Call the vertices with outdegree 0 the sinks of $G$. Let $|V|= N$. We describe the \emph{stone formula on $G$ twisted with $\rho$}, $\Stone(G, \rho)$ below. In words, the contradiction we are about to describe states the following:
\begin{itemize}
    \item There are $|V|$ stones. Each stone has a color: red or blue.
    \item At least one stone must be placed on each vertex.
    \item All stones placed on sinks must be red.
    \item All stones placed on the root must be blue.
    \item Let $v$ be a node with out-neighbors $u,w$. If a red stone $j$ is placed on $u$ and a red stone $k$ is placed on $w$, then all stones placed on $v$ must be red.
\end{itemize}
We shall twist this formula with an obfuscation map $\rho$ to make it hard for regular resolution. We formally define the formula below.
We introduce the following set of variables.
\begin{description}
    \item[Vertex variables:] For all $ v \in V, 1 \leq j \leq N: P_{v,j}$. \newline
    Semantic interpretation: $P_{v,j}$ is set to 1 iff stone $j$ is placed on vertex $v$.
    \item[Stone variables:] For all $1 \leq j \leq N: R_j$. \newline
    Semantic interpretation: $R_j$ is set to 1 if stone $j$ is colored red, otherwise it is set to 0.
    \item[Auxiliary variables:] For all $v \in V, 1 \leq j \leq N - 1 : Z_{v,j}$. \newline
    Semantic interpretation: These are auxiliary variables used to encode the fact that at least one stone is placed on vertex $v$, with a bunch of constant-width clauses.
\end{description}
Let $\V$ denote the set of all variables mentioned above and $\rho: [N]^3 \to \V$ be an arbitrary mapping that we call an \emph{obfuscation map}.
Let $S$ be the set of sinks of $G$. We define $\Stone(G, \rho)$ to be the formula comprising the following set of clauses:
\begin{description}
    \item[Root clauses:] For all $1 \leq j \leq N$: $\neg P_{r,j} \lor \neg R_j$ \newline
    Semantic interpretation: All stones placed on the root $r$ of $G$ must be coloured blue.
    \item[Sink clauses:] For all $1 \leq j \leq N, s \in S$: $\neg P_{s,j} \lor R_j$ \newline
    Semantic interpretation: Each stone placed on a sink of $G$ must be coloured red.
    \item[Induction clauses:] For all $v \in V(G)$ with out-neighbors $u, w$ and for each $i,j,k \in [N]$:
    \[
    \neg P_{u, i} \lor \neg R_{i} \lor \neg P_{w, j} \lor \neg R_{j} \lor \neg P_{v,k} \lor R_k \lor \rho(i, j, k)
    \] 
    \[
    \neg P_{u, i} \lor \neg R_{i} \lor \neg P_{w, j} \lor \neg R_{j} \lor \neg P_{v,k} \lor R_k \lor \neg \rho(i, j, k)
    \]       
    Semantic interpretation: after resolving out the variable $\rho(i, j, k)$, the clause says that if the stones placed on $u$ and $w$ are colored red, the stone placed at $v$ must also be colored red, i.e., the implication
    \[
    (P_{u, i} \land R_{i} \land P_{w, j} \land R_{j} \land P_{v,k}) \Longrightarrow R_k.
    \]
    \item[Stone-placement clauses:] For all $v \in V(G)$:
    \begin{align*}& P_{v,1} \lor \neg Z_{v,1} \\
     &Z_{v,1} \lor P_{v,2} \lor \neg Z_{v,2} \\
     &\dots\\ 
     &Z_{v, N-2} \lor P_{v, N-1} \lor \neg Z_{v, N-1}\\
     &Z_{v, N-1} \lor P_{v,N}\end{align*}
     Semantic interpretation: Together, the clauses are equivalent to
     $$ P_{v,1} \lor P_{v,2} \lor \cdots \lor P_{v,N}$$
     i.e., at least one stone is placed on the vertex $v$.
\end{description}


Let $G_n$ be the pyramid graph on $n$ levels. The vertex set is $V=\{ (i,j)| 1 \leq i \leq n, 1 \leq j \leq i\}$. The \emph{level} of a vertex $(i,j)$ is defined to be its first coordinate $i$. The edge set is $E= \{ (i,j) \rightarrow (i+1, j) | 1 \leq i < n, 1 \leq j \leq i\} \cup \{ (i,j) \rightarrow (i+1, j+1) | 1 \leq i < n, 1 \leq j \leq i\}$. 
See Figure~\ref{fig:PyramidGraph}, for an example of the pyramid graph.
The sinks of $G_n$ are the vertices at layer $n$, i.e., $(n,i)$ for $1 \leq i \leq n$. The root is $(1,1)$. 
We have $|V| = N = \frac{1}{2}n(n+1)$.

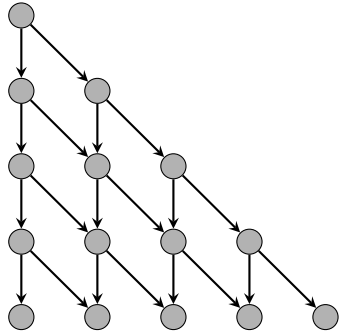
\begin{figure}[ht]
    \centering
    
\begin{tikzpicture}
\node (1001)[blackcircle]{};
\node (2001)[blackcircle, below of = 1001]{};
\node (2002)[blackcircle, right of = 2001]{};
\node (3001)[blackcircle, below of = 2001]{};
\node (3002)[blackcircle, right of = 3001]{};
\node (3003)[blackcircle, right of = 3002]{};
\node (4001)[blackcircle, below of = 3001]{};
\node (4002)[blackcircle, right of = 4001]{};
\node (4003)[blackcircle, right of = 4002]{};
\node (4004)[blackcircle, right of = 4003]{};
\node (5001)[blackcircle, below of = 4001]{};
\node (5002)[blackcircle, right of = 5001]{};
\node (5003)[blackcircle, right of = 5002]{};
\node (5004)[blackcircle, right of = 5003]{};
\node (5005)[blackcircle, right of = 5004]{};
\draw [arrow](1001) -- (2002);
\draw [arrow](1001) -- (2001);
\draw [arrow](2001) -- (3002);
\draw [arrow](2001) -- (3001);
\draw [arrow](2002) -- (3003);
\draw [arrow](2002) -- (3002);
\draw [arrow](3001) -- (4002);
\draw [arrow](3001) -- (4001);
\draw [arrow](3002) -- (4003);
\draw [arrow](3002) -- (4002);
\draw [arrow](3003) -- (4004);
\draw [arrow](3003) -- (4003);
\draw [arrow](4001) -- (5002);
\draw [arrow](4001) -- (5001);
\draw [arrow](4002) -- (5003);
\draw [arrow](4002) -- (5002);
\draw [arrow](4003) -- (5004);
\draw [arrow](4003) -- (5003);
\draw [arrow](4004) -- (5005);
\draw [arrow](4004) -- (5004);

\end{tikzpicture}
    \caption{An example of the pyramid graph with $n=5$ levels.}
    \label{fig:PyramidGraph}
\end{figure}

We instantiate the stone formula with $G=G_n$. Let $\SP_{n,\rho} = \Stone(G_n, \rho)$.
We denote the number of variables of $\SP_{n,\rho}$ by $m$, i.e. $|\V| = m = N^2 + N + N(N - 1) = 2N^2$.
In order to prove our lower bound against regular ResLin, it turns out to be convenient working with a formula that is obtained by \emph{lifting} $\SP_{n, \rho}$.
Let $g: \cube^b \rightarrow \cube$ be a Boolean function, called gadget. For $a \in \{0,1\}$, we denote the set of all pre-images of $a$ by $g^{-1}(a)$, i.e., $g^{-1}(a) \coloneqq \bigl\{ x \in \cube^b \mid g(x) = a \bigr\}$.

Let $c_1, c_2, \cdots , c_k \in \{0,1\}$. Let $C = [X_1 = c_1]  \lor \cdots \lor [X_k = c_k]$ be a clause over variables $X_1,\dots,X_k$. Here $[X_i = c_i]$ denotes a literal, i.e. $X_i$ if $c_i = 1$, and $\neg X_i$ otherwise.
To lift Clause $C$ we introduce $b$ variables $Y^i_1,\dots,Y^i_b$ for each variable $X_i$ of $C$.
The lift of $C$, $C \circ g$, is a set of clauses which, in conjunction, are semantically equivalent to $ [g(Y^1_1, \cdots, Y^1_b) = c_1] \lor \cdots \lor [g(Y^k_1, \cdots , Y^k_b)=c_k] $, i.e., the following:


\[
C \circ g \coloneqq \left\{  \bigvee_{1 \leq i \leq k, 1 \leq j \leq b} [Y^i_j =  1 - d^i_j]  ~\Bigl|~ d^1 \in g^{-1}(1 - c_1), \dots , d^k \in g^{-1} (1 - c_k) \right\},
\]
where each $d^i$ is a $b$-bit string and $d^i_j$ is its $j$-th bit.

\begin{observation}
\label{obs:LiftedClause}
    An assignment $(Y^i_j)_{1 \leq i \leq k, 1 \leq j \leq b}$ satisfies every clause in $C \circ g$ if and only the lifted assignment $(x_1, x_2, \cdots , x_k)$ given by 
    $$x_1 = g (Y^1_1, \cdots , Y^1_b), \cdots , x_k = g(Y^k_1, \cdots , Y^k_b)$$
    satisfies the clause $C$.
\end{observation}

For a technical reason, we shall also need the following simple observation:
\begin{observation}
\label{obs:AllVariablesInLift}
    Let $\psi$ be any clause of $C \circ g$. Suppose, one of the variables $Y^i_j$ appears in $\psi$. Then, for all $k \in [b]$, the variable $Y^i_k$ also appears in $\psi$.
\end{observation}

For a set of clauses $\Phi$ (a CNF formula), we define its lift as $ \Phi \circ g \coloneqq \cup_{C \in \Phi} (C \circ g)$.
We have the following corollary of Observation~\ref{obs:LiftedClause}.
\begin{corollary}
The set of clauses $\Phi$ is unsatisfiable if and only if the set of clauses $\Phi \circ g$ is unsatisfiable.
\end{corollary}

We remark that if the base set $\Phi$ contains only clauses of width at most $k$, then $\Phi \circ g$ contains clauses of width at most $kb$ and $|\Phi \circ g| \leq 2^{bk}\cdot|\Phi|$.
In particular, a constant-width, poly-size unsatisfiable formula, when lifted by an $O(\log n)$ size gadget, yields an $O(\log n)$-width, poly-size unsatisfiable formula.
Our hard formula will be the stone formula lifted by an inner product gadget $\SP_{n,\rho} \circ \IP$, where $\IP: \cube^b \to \cube$ is the inner-product function, i.e. for each $x,y \in \{0,1\}^{b/2}$, set $\IP(x,y) = x_1y_1 + \cdots + x_{b/2}y_{b/2} \text{ (mod }2\text{)}$ where $b = O(\log n)$ is an even integer.
\arkadev{Here, instead of writing $O(\log n)$, we should write it $\theta\cdot\log n$, where $\theta$ is some constant. Would be nice to fix the value of $\theta$.}

\paragraph{Comparison with related formulas from previous work:}  We briefly mention some formulas that were used in the past, that are related to $\Stone(G_n,\rho)$. First, the formula sans obfuscation $\Stone(G_n)$ was one of the formulas used by Alekhnovich et al. \cite{AJPU02} to yield the first exponential separation between regular and general resolution. This separation was further strengthened by Urquhart \cite{U11} in a follow-up work, using the following more involved formula. Fix a bijective placement of stones to the vertices of $G_n$. This reduces $\Stone(G_n)$ to a plain pebbling formula on pyramid graph, denoted by $\text{Peb}(G_n)$. This reduces the number of clauses to linear in $N$. Urquhart considers the 2-bit XOR lift of such a formula, i.e. $\text{Peb}(G_n) \circ \oplus$. This formula blows up the number of clauses, but still keeps the number of variables to $O(N)$. Finally, he shows that for a suitable $\rho$, the obfuscation of $\text{Peb}(G_n) \circ \oplus$ by $\rho$, just as we obfuscate $\Stone(G_n)$ by $\Stone(G_n,\rho)$, yields a separation between regular and general resolution that is stronger than the one by Alekhnovich et al. \cite{AJPU02}. More recently, Vinyals et al. \cite{VEJN20} worked with a different \emph{sparsification} of $\Stone(G_n)$. This comes about naturally by considering $\Stone(G_n)$ as a densification of $\text{Peb}(G_n)$ by using a complete bi-partite graph with $N$ vertices on each side. Roughly speaking, Vinyals et al. used a constant degree bi-partite expander gadget with $\text{Peb}(G_n)$, inspired by the earlier work of Razborov \cite{R16}. This results in more modular and optimal arguments. However, the lift by a stifled gadget of a base stone formula, like we do as in $\Stone(G_n,\rho) \circ \IP$, seems not to have been considered earlier.

\section{Resolution Proof Systems and Branching Programs}
\label{sec:BP_Res}
A proof in a propositional proof system starts from a set of clauses $\Phi$, called axioms, that is purportedly unsatisfiable.  It generates a proof by deriving the empty clause from the axioms, using inference rules. The main inference rule in the standard resolution, called the resolution rule,  derives  a clause $A \vee B$  from clauses $A \vee x$ and $B \vee \lnot x$ (i.e., we resolve the variable $x$).
If we can derive the empty clause from the original set $\Phi$ then it proves the set $\Phi$ is unsatisfiable.
We will need the following basic and well known fact  that states resolution is complete without being too inefficient.

\begin{lemma} \label{lemma:resolution-complete}
    Let $C$ be any clause, and $\Phi$ be any CNF formula over $n$ Boolean variables and of polynomial size, that semantically implies $C$. Then, $C$ can be derived from $\Phi$ by a resolution proof of size at most $2^{O(n)}$. 
\end{lemma}
\pavel{Put a reference for this lemma?}

Linear resolution, aka ResLin and introduced by Raz and Tzamaret \cite{RT08},  is a generalization of standard resolution, using linear clauses (disjunction of linear equations over $\F_2$) to express lines of a proof.
It consists of two rules:
\begin{description}
    \item[Resolution Rule:] From linear clauses $A \vee (\ell = 0)$ and $B \vee (\ell = 1)$ derive a linear clause $A \vee B$.
    \item[Weakening Rule:] From a linear clause $A$ derive a linear clause $B$ that is semantically implied by $A$ (i.e., any assignment satisfying $A$ also satisfies $B$).
\end{description}
The length of a resolution (or ResLin) refutation of a formula $\Phi$ is the number of applications of the rules above in order to refute the formula $\Phi.$
The width of a resolution (or ResLin) refutation is the maximum width of any (linear) clause that is used in the resolution proof.

It is known that a resolution proof and a linear resolution proof, for an unsatisfiable set of clauses $\Phi$, correspond to a branching program and a linear branching program, respectively, for a search problem $\Search(\Phi)$ (see for example Garg et al \cite{GGKS20}, who credit it to earlier work of Razborov \cite{Razborov95} that was simplified by Pudl{\'a}k \cite{Pudlak10} and Sokolov\cite{Sokolov17}) 
that is defined as follows.
For a given assignment $\alpha$ of the $n$ variables of $\Phi$, one needs to find a clause in $\Phi$ that is unsatisfied by $\alpha$ (at least one exists as the set $\Phi$ is unsatisfiable).
A linear branching program computing a search problem $\Prob \subseteq \F_2^n \times O$ is defined as follows.
\begin{itemize}
    \item There is a directed acyclic graph $\mathcal{P}$ of one source and some sinks. Each non-sink node has out-degree at most two.
    For an inner node $v$ the two out-neighbors $u$ and $w$ (i.e., there are edges $(v,u)$ and $(v,w)$ in $\mathcal{P}$) are called children of $v$.
    \item Each node $v$ of $\mathcal{P}$ is labeled by an affine space $A_v \subseteq \F^n_2$.
    \item The source is labeled by $\F^n_2$. 
    \item Let $v$ be a node of out-degree 2 and $u$ and $w$ be children of $v$. Then, $A_u = A^0_v$ and $A_w = A^1_v$, where $A^c_v = \{ x \in A_v \mid \langle f_v, x \rangle = c\}$ for a linear query $f_v = \F^n_2$ and $c \in \{0,1\}$. We call such $v$ a \emph{query node}.
    \item Let $v$ be a node of out-degree 1 and $u$ be the child of $v$. Then, $A_v \subseteq A_u$. We call such $v$ a \emph{forget node}.
    \item Each sink $v$ of $\mathcal{P}$ has an assigned output $o_v \in O$ such that $A_v$ is $o_v$-monochromatic according to $\Prob$, i.e., $\alpha \in A_v \implies (\alpha,\,o_v) \in \Prob$.
\end{itemize}

A standard/ordinary branching program is defined analogously but its nodes are labeled by cubes instead of affine spaces. Consequently, variables instead of arbitrary linear functions are queried at its query nodes.

The correspondence between a branching program computing $\Search(\Phi)$ and a (linear) resolution proof refuting $\Phi$ is roughly the following.
We can represent the resolution proof as a directed acyclic graph where nodes are labeled by (linear) clauses.
The sources are labeled by clauses of $\Phi$ and there is exactly one sink that is labeled by an empty clause. 
Each node that is not a source has at most two parents and it corresponds to an application of the (linear) resolution rule (if the node has 2 parents), or the weakening rule (if the nodes has 1 parent). 
To get a (linear) branching program for $\Search(\Phi)$ we just flip the direction of the edges in the resolution graph and negate the clauses that are used for node labeling.
Thus, each node is labeled by a cube or an affine space, the query nodes correspond to applications of the resolution rule, and the forget nodes correspond to applications of the weakening rule.
It is clear the size of a branching program $\cP$ (number of nodes of $\cP$) is exactly the same as length of the corresponding resolution refutation.

Regular resolution is a subsystem of the resolution system, such that in any path of the resolution proof graph each variable can be resolved at most once.
A read-once branching program corresponds to a regular resolution proof, i.e., on each directed path from the source to a sink each variable is queried at most once.
There is interest in two generalizations of regular resolution to linear regular resolution -- top-regular linear resolution~\cite{GPT22} and bottom-regular linear resolution~\cite{EGI23} (in both papers called as regular linear resolution).
We will define both of them by their corresponding linear branching programs.

\begin{definition}[\cite{GPT22}]
Let $v$ be a node of a linear branching program $\mathcal{P}$.
Let $\Pre(v)$ be the space spanned by all linear functions queried on any path from the source of $\mathcal{P}$ to $v$.
Let $\Post(v)$ be the space spanned by all linear functions queried on any path from $v$ to any sink of $\mathcal{P}$.
\end{definition}

A linear branching program is \emph{top-read-once}\footnote{Gryaznov et al.~\cite{GPT22} used just the name weakly read once for such programs.}~\cite{GPT22} if for each query node $v$, we have $f_v \not \in \Pre(v)$.
A linear branching program is \emph{bottom-read-once}~\cite{EGI23} if for each edge $(v,u)$ such that $v$ is a query node holds that $f_v \not \in \Post(u)$.
A linear resolution proof is \emph{top-regular}, or \emph{bottom-regular} if the corresponding branching program is top-read-once, or bottom-read-once, respectively.
\arkadev{Two things: perhaps it is better to hyphenate the whole thinf, i.e. say bottom-read-once instead of bottom read-once. Ditto for top-read-once. Second, we perhaps should mention what would be a read-once linear branching program. This should be a common generalization of both bottom and top read-once.}
We use both notion of branching program and resolution to state and prove our result, whichever is more suitable for the presentation at hand.
Our separation is only for bottom-regular ResLin, i.e., for bottom-read-once linear branching program, which we abbreviate to BROLBP.

\begin{lemma}[Lemma 2.6~\cite{EGI23} stated for branching programs]
\label{lem:PostDim}
Let $\cP$ be a BROLBP computing a search problem $\Prob \subseteq \{0,1\}^n \times O$.
Let $v$ be a node of $\cP$ such that there is a path of length $t$ from the source of $\mathcal{P}$ to $v$.
Then, $\dim(\Post(v)) \leq n - t$.
\end{lemma}

\begin{lemma}[Lemma 2.3~\cite{EGI23} stated for branching programs]
\label{lem:BPPath}
Let $\cP$ be a linear branching program computing a search problem $\Prob \subseteq \{0,1\}^n \times O$.
Let $u$ and $v$ be nodes of $\cP$ such that there is a directed path $p$ from $u$ to $v$.
Let $(M|c)$ be the system of linear equations given by queries along the path $p$ and $A$ be the affine space of solution of $(M|c)$, i.e., $A = \{\alpha \in \{0,1\}^n \mid M\alpha = c\}$.
Then, for $A_u$ and $A_v$ the affine spaces associated with $u$, and $v$, respectively, holds that $A_u \cap A \subseteq A_v$.
\end{lemma}

Even more restrictive subsystems, are tree-like resolution, and tree-like ResLin.
These subsystems correspond to decision trees and parity decision trees.
A parity decision tree (PDT) is a linear branching program such that its underlying graph is a tree, and a decision tree (DT) is a restriction where we query only bits of the input, instead of linear functions.
It is clear that tree-like resolution is a subsystem of regular resolution.
Analogously, tree-like ResLin is a subsystem of both bottom-regular and top-regular ResLin.

It is easy to see that strongly read-once linear branching programs are both top-read-once and bottom-read-once. Chattopadhyay and Liao \cite{CL23} showed that strongly read-once linear branching programs can simulate parity decision trees. A super-polynomial separation between tree-like ResLin and bottom-regular ResLin follows from the lifting theorems in Chattopadhay et al \cite{CFKMP21} as well as in Beame and Koroth \cite{BK23}. We have the following containments:
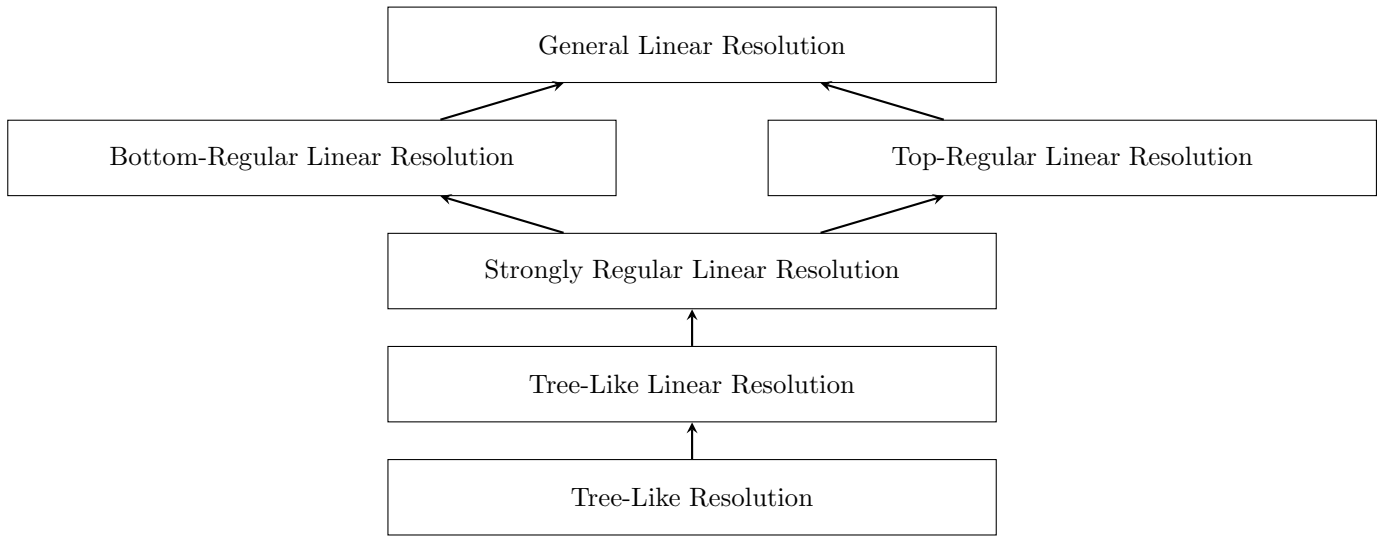
\begin{center}
\begin{tikzpicture}
\node (DT) [process] {Tree-Like Resolution};
\node (PDT) [process, above of= DT, yshift= 0.5cm] {Tree-Like Linear Resolution};
\node (SROLBP) [process, above of= PDT, yshift= 0.5cm] {Strongly Regular Linear Resolution};
\node (bottom-regular)[process, above of= SROLBP, yshift= 0.5cm, xshift=-5cm]{Bottom-Regular Linear Resolution};
\node(top-regular)[process, above of= SROLBP, xshift=5cm, yshift=0.5cm]{Top-Regular Linear Resolution} ;
\node(ResLin)[process, above of= top-regular, xshift=-5cm, yshift=0.5cm]{General Linear Resolution};
\draw [arrow] (DT)--(PDT);
\draw [arrow] (PDT)--(SROLBP);
\draw [arrow] (SROLBP) -- (bottom-regular);
\draw [arrow] (SROLBP) -- (top-regular);
\draw [arrow] (bottom-regular) -- (ResLin);
\draw [arrow] (top-regular) -- (ResLin);
\end{tikzpicture}

\end{center}

In this paper we show the existence of a CNF formula with a polynomial sized resolution refutation for which any bottom-regular linear resolution refutation requires exponential size. Thus, we show that the containment \textit{bottom-regular linear resolution} $\subseteq$ \textit{general ResLin} is strict. 

\subsection{ROLBP Computing Boolean Function}  \label{sec:BooleanFunction}
In the context of computing Boolean functions, branching programs are defined, usually, in a more relaxed fashion in a certain sense. For instance, ordinary branching programs are defined without placing the restriction that the set of inputs reaching a node can be contained in a non-trivial sub-cube. This is something we insist when we define BPs here as our focus is on capturing the limitations of those BPs that are derived from a resolution proof DAG by reversing the direction of its edges. It, possibly, would have been more meaningful to call these latter objects affine DAGs, but we chose to call it linear branching programs for the sake of continuity wrt the earlier works by Efremenko et al. \cite{EGI23} and Gryaznov et al. \cite{GPT22}. In this section, we take the liberty of indeed calling them DAGs to compare them with BPs. Affine DAGs severely restrict the power of computing Boolean functions. This is because the set of sink nodes of such a DAG of small size computing a Boolean function $f$ simply provides an efficient affine cover of $f^{-1}(0)$ as well as $f^{-1}(1)$. Thus, immediately, one concludes that any affine DAG, without any restriction on the number of reads of a variable, computing the Inner-product on $n$ bits requires $2^{\Omega(n)}$ size\footnote{In fact, the case for cube DAGs is known to be more dramatic. If a Boolean function $f$ can be computed by a cube DAG of size $s$, then it can be also computed by a decision tree of size $s^{O(\log(s)\cdot\log n)}$.} as $\IP$ has large affine cover number. On the other hand, $\IP$ can be easily seen to be computed by a linear-size read-once and bit-querying branching program. 

On the other hand, the situation for problems of searching a falsified clause, is quite different. As Lemma 2.4 in the work of Efremenko et al. \cite{EGI23} proves, for any unsatisfiable CNF $\Psi$, any top-read-once linear branching program solving $\Search\big(\Psi\big)$ gives rise to a top-read-once affine DAG for $\Search\big(\Psi\big)$ with hardly any blow-up in its size. The proof can be easily verified to additionally yield that a strongly read-once linear BP for $\Search\big(\Psi\big)$, completely analogously, yields a strongly read-once affine DAG for $\Search\big(\Psi\big)$ with no essential blow-up to its size. Thus, our main result, Theorem~\ref{thm:MainResult} holds equally for strongly read-once linear branching programs that are not restricted by definition to be affine DAGs.

\section{Linear Algebraic Facts} \label{sec:linear}
In this section, we will describe the notions of linear algebra that we will need in our arguments. Let us introduce some notation first. 
Let $M \in \F^{t \times m}_2$ be a matrix. We denote the row space of $M$ by $\Row(M)$. For a vector $c \in \{0,1\}^t$, $\Space(M,c)$ is the affine space of solutions to the linear system $(M|c)$, i.e., $\Space(M,c) = \{ \alpha \in \{0,1\}^m \mid M\cdot \alpha = c\}$. 

The entries of vectors of $\F^{mb}_2$ are naturally divided into $m$ blocks, each having $b$ co-ordinates/bits, i.e., for $j \in [m]$, the $j$-th block contains the coordinates $(j-1)b + 1, \dots, jb$. 
For $j \in [m]$, $\text{BLOCK}(j) = \{ (j-1)b+1, \cdots , jb \}$. 
Also for $T \subseteq [m]$ define $\text{BLOCK}(T) = \cup_{t \in T} \text{BLOCK}(t)$. 
For a vector $u \in \F^{mb}_2$ and a block $j \in [m]$, $u^j \in \F^b_2$ is the vector corresponding to the block $j$ of $u$, i.e., $u^j = (u_{(j-1)b+1}, \dots, u_{jb})$.
We say a vector $u \in \F^{mb}_2$ \emph{touches} a block $j \in [m]$ if the vector $u^j$ is non-zero. 
A set of vectors $R \subseteq \F^{mb}_2$ touches a block $j$ if at least one of the vectors in $R$ touches $j$.
Let $U$ be a subspace of $\F^{mb}_2$ and $T \subseteq [m]$ be a set of blocks. 
The subspace $U_T$ of $U$ is the linear space of all vectors $u$ that do not touch any block outside $T$, i.e., $U_T = \{u \in U \mid \forall j \not \in T: u^j = (0,\dots, 0) \}$.
For $S = \text{BLOCK}(T)$, the subspace $U_{\downarrow T}$ of $\F^S_2$ is the projection of $U$ onto $T$, i.e., $ U_{\downarrow T} = \{x \in \F^{S}_2 \mid \exists y \in \F^{[mb] \setminus S}_2: (x,y) \in U \}$.
We call a tuple of vectors $R = (u_1, \dots , u_t), u_i \in \F^{mb}_2$ to be \emph{safe} if the following condition holds:
\begin{itemize}
    \item The vectors $(u_1, \dots , u_t)$ form a matrix $M \in \F^{t \times mb}$ in echelon form, i.e., there are $t$ distinct coordinates $a_1, \dots , a_t \in [mb]$ such that for all $i, j \in [t]$:

\begin{align*}
    (u_i)_{a_j} = \begin{cases}
        1 & \text{ if } i=j \\
        0 & \text{ otherwise}
    \end{cases}
\end{align*}
In other words, the matrix $M$ restricted to the columns $a_1,\dots,a_t$ is the identity matrix $I_t \in \F^{t \times t}_2$.

\item Moreover, each pivot $a_i$ lies in a distinct block.
\end{itemize}

The $a_i$'s are called the pivot variables of $R$. 
There might be multiple possible choices for the tuple of pivot variables $(a_1, \dots , a_t)$. In that case we pick any valid choice, say the lexicographically smallest valid choice, and call it the set of pivots.

A subspace $U$ of $\F^{mb}_2$ is \emph{spread} if any set of $k$ linearly independent vectors of $U$ touches at least $k$ blocks, for each $1 \le k \le m $.
We say a set of blocks $T \subseteq [m]$ is an \emph{obstruction} of a space $U$ if $U_{\downarrow \bar{T}}$ is spread, where $\bar{T}$ is complement of $T$, i.e., $\bar{T} = [m] \setminus T$. An obstruction $T \subseteq [m]$ of a space $U$ is minimal if any proper subset $T' \subset T$ is not an obstruction of $U$, i.e., $U_{\downarrow \bar{T'}}$ is not spread.
Efremenko et al.~\cite{EGI23} showed the following result.

\begin{theorem}[Theorem 3.1, Efremenko et al.~\cite{EGI23}]
\label{thm:SafeBasis}
    Let $U$ be a spread subspace of $\F^{mb}_2$ with $\text{dim}(U) \leq m$. Then, $U$ has a safe basis.
\end{theorem}

 The following is a basic fact.
\begin{observation}  \label{obs:linear-decomposition}
Let $U$ be any subspace of $\mathbb{F}_2^{mb}$ and $T \subseteq [m]$ be a set of blocks. Then, $$\text{dim}(U) = \text{dim}(U_T) + \text{dim}\big(U_{\downarrow \bar{T}}\big).$$
\end{observation}

\begin{proof}
    Let $U'$ be a suitable subspace of $U$ such that $U = U_T \oplus U'$. It is simple to see that $U_{\downarrow \bar{T}} = (U')_{\downarrow \bar{T}}$. This is because every vector $u \in U$ can be written as $x + u'$, with $x \in U_T$ and $u' \in U'$. But, as $x_{\downarrow \bar{T}} = 0$, we have $u_{\downarrow \bar{T}} = u'_{\downarrow \bar{T}}$. Hence, we conclude that $\text{dim}\big( U_{\downarrow \bar{T}}\big) = \text{dim}\big( (U')_{\downarrow \bar{T}}\big) \le \text{dim}\big(U'\big)$.
    To establish our result, we will simply show that $\text{dim}(U') \le \text{dim}\big( (U')_{\downarrow \bar{T}}\big)$. This follows if we show that whenever $u'_1,\ldots,u'_r \in U'$ are linearly independent vectors, so are $\big(u'_1\big)_{\downarrow \bar{T}}, \ldots, \big(u'_r\big)_{\downarrow \bar{T}}$. If that is not the case then there exists a vector $x \in U_T$ such that $x,u'_1,\ldots,u'_r$ are not linearly independent, contradicting our assumption. 
\end{proof}


Efremenko et al.~\cite{EGI23} showed the following properties of minimal obstructions.
\begin{lemma}
\label{lem:Closure}
    Let $U$ be a subspace of $\F^{mb}_2$.
    Then, a minimal obstruction $T \subseteq [m]$ of $U$ is unique and $|T| \leq \dim(U)$.
\end{lemma}

\begin{definition}
    For an affine space $A= \Space(M,c) \subseteq \F^{mb}_2$ we define its closure $\Cl(A) \subseteq [m]$ to be its unique minimal obstruction.  Also, define $\VarCl(A) \subseteq [mb]$ to be the set of variables that appear in the blocks of $\Cl(A),$ i.e., 
    $$ \VarCl(A)= \text{BLOCK}(\Cl(A)).$$
\end{definition}

Efremenko et al.~\cite{EGI23} proved the following relationship between the closures of two affine spaces when one contains the other. 
\begin{lemma}
\label{lem:MonotoneClosure}
 Let $A \subseteq A'$ be two affine spaces of $\F^{mb}_2$.
 Then, $\Cl(A') \subseteq \Cl(A)$.
\end{lemma}

A \emph{partial assignment} $\alpha'$ of $m$ variables is a string in $\{0,1,*\}^{m}$. A variable $X \in [m]$ is \emph{assigned} if $\alpha_X \in \{0,1\}$.
For a total assignment $\alpha \in \{0,1\}^m$ and $T \subseteq [m]$ we define the restriction $\alpha_{|T}$ of $\alpha$ to $T$ to be the partial assignment arising from $\alpha$ by unassigning the variables that are not in $T$, i.e., for each $i \in [m]$
\[
\bigl(\alpha_{|T})_i = 
\begin{cases}
    \alpha_i & \text{if $i \in T$}, \\
    * & \text{otherwise.}
\end{cases}
\]

We describe the notion of \emph{stifling} introduced by Chattopadhay et al.~\cite{CMSS23}.  
\begin{definition}
    A Boolean function $g: \{0,1\}^b \to \{0,1\}$ is \emph{stifled}\footnote{1-stifled called by Chattopadhyay et al.~\cite{CMSS23}} if the following holds
    \begin{align*}
        &\forall i \in [b]  \text { and } a \in \{0,1\} ~ \exists \delta \in \{0,1\}^{b}  \\
        & \text{such that for all }\gamma \in \{0,1\}^b \text{ with } \gamma_{|[b]\setminus \{i\}} = \delta_{|[b]\setminus \{i\}} \text{ holds that } g(\gamma) = a.
    \end{align*}
\end{definition}
We call $\delta$ from the previous definition a \emph{stifling assignment} for $i$ and $a$.
The utility of stifling is the following.
An adversary can pick any variable $i \in [b]$ of $g$.
For any $a \in \{0,1\}$, we can pick a partial assignment $\delta_a \in \{0,1,*\}^m$ that assigns a value to all variables except the $i$-th variable.
Now, no matter how the adversary chooses the value for the $i$-th variable to get a total assignment $\gamma_a \in \{0,1\}^b$ from $\delta_a$, the value $g(\gamma_a)$ will be always $a$.

\begin{definition}
A partial assignment $\beta \in \{0,1,*\}^{mb}$ is called \emph{block-respecting} if for each block $j \in [m]$, either all variables or no variables are assigned, i.e.,
$$ (\beta^j)_i \in \{0,1\} \text{ for all } i \in [b] \text{ or } (\beta^j)_i = * \text{ for all } i \in [b].$$
\end{definition}

A block-respecting assignment $\beta \in \{0,1,*\}^{mb}$ naturally gives a partial assignment $\overrightarrow{g}(\beta) \in \{0,1,*\}^m$ by applying the gadget $g$ to the assigned blocks.
Formally, for each $j \in [m]$ we have
\[
\overrightarrow{g}(\beta)_j = 
\begin{cases}
    g(\beta^j_1,\dots,\beta^j_b) & \text{if for all $i \in [b]:\beta^j_i$ are assigned,} \\
    * & \text{otherwise.}
\end{cases}
\]

\begin{definition}
    Let $A \subseteq \F^{mb}$ be an affine space and $\beta \in A$. The \emph{closure-assignment} of $\beta$, $\beta|_{\VarCl(A)}$ is the partial assignment which fixes all coordinates in blocks of $\Cl(A)$ according to $\beta$ and keeps other coordinates free. In other words, 
    \begin{align*}
        (\beta|_{\VarCl(A)})^j = \begin{cases}
            \beta^j & \text { if } j \in \Cl(A), \\
            (*, \dots  , *) & \text{ otherwise}.
          \end{cases} 
    \end{align*}
\end{definition}

\begin{lemma}
\label{lem:StiflingGadget}
Let $A = \Space(M,c) \subseteq \F^{mb}$ be an affine space and let $g: \{0,1\}^b \to \{0,1\}$ be a stifled gadget.
Let $\beta \in A$ be a vector and $\beta' \in \{0,1,*\}^{mb}$ be its closure assignment. Let $\alpha' := \overrightarrow{g}(\beta') \in \{0,1,*\}^m$.
Then, for any extension of $\alpha'$ to a total assignment $\alpha \in \{0,1\}^m$, there exists $\gamma \in A$ such that $\overrightarrow{g}(\gamma) = \alpha$.
\end{lemma}
\begin{proof}

    WLOG assume the rows of $M$ are linearly independent. Let $U = \Row(M)$ be the row-space of $M$ and let $T \subseteq [m]$ be the closure of $A$. First, we construct a matrix $M'$ which has the same row-space as $M$. 
    
    \paragraph{Construction of $M'$}
    \begin{enumerate}

   \item Let $(u_1,\dots,u_d)$ be an arbitrary basis of $U_T$ and let $M_1 \in \F^{d \times mb}$ be the matrix whose rows are the vectors $u_1,\dots,u_d$:
\begin{align*}
    M_1 & =  \left[\begin{array}{c}
        u_1 \\
        u_2 \\
        \vdots \\
        u_d
    \end{array}\right]
\end{align*}
 
    \item Let $(w_1,\dots,w_{d'})$ be a safe basis of $U_{\downarrow \bar{T}}$.
    Such a basis exists by the definition of closure and Theorem~\ref{thm:SafeBasis}.
    Let $a_1,\dots,a_{d'}$ be pivots of $w_1,\dots,w_{d'}$.
    Each of these pivots lie in a distinct block.
    Moreover, none of these pivots are in the blocks of $T$.
    
    Let $L: U \to U_{\downarrow \bar{T}}$ be the projection of $U$ to $U_{\downarrow \bar{T}}$.
    Let $w'_i$ be an arbitrary pre-image of $w_i$ according to $L$, i.e., $L(w'_i) = w_i$.
    Since $(w_1,\dots,w_{d'})$ are linearly independent, the vectors $(w'_1,\dots,w'_{d'})$ are linearly independent as well.
    Let $M_2 \in \F^{d' \times mb}$ be the matrix with the vectors $w'_1,\dots,w'_{d'}$ as its rows.
    \begin{align*}
        M_2 & = \left[ \begin{array}{c} 
        w'_1 \\
        w'_2 \\
        \vdots \\
        w'_{d'}
        \end{array} \right]
    \end{align*}
    \item Take $M' \in \F^{(d+d') \times mb} $ to be the matrix obtained by stacking $M_1$ on top of $M_2$:

    \begin{align*}
    M' &= \left[ \begin{array}{c} M_1 \\ M_2 \end{array} \right] 
    \end{align*}

    \end{enumerate}
\begin{claim}
\label{claim:same_row_space}
    The matrices $M$ and $M'$ have the same row-space.
\end{claim}
\begin{proof}
By Observation~\ref{obs:linear-decomposition}, $\dim(U) = \dim(U_T) + \dim(U_{\downarrow\bar{T}}) = d+d'$. No row of $M_2$ can be generated by rows of $M_1$ as the pivots $a_1,\dots,a_{d'}$ of the matrix $M_2$ lie in columns where the matrix $M_1$ has only 0 entries.
Thus, $\rank(M') = \rank(M_1) + \rank(M_2) = d+d'= \dim(U) = \rank(M)$.
Moreover, any row of $M'$ lies in $U = \Row(M)$. It follows that $\Row(M) = \Row(M')$.
\end{proof}

Thus, there is a vector $c' \in \F^{mb}$ such that $A = \Space(M',c')$. 
Now, we prove the lemma. We are given a vector $\beta \in \Space(M',c')$ and a target assignment $\alpha \in \{0,1\}^m$ such that $g(\beta^j)= \alpha_j$ for all $j \in \Cl(A)$. Our goal is to show the existence of a $\gamma \in \F^{mb}$ such that $\overrightarrow{g}(\gamma) = \alpha$ and $M' \gamma= c' $. 
Recall that the tuple of rows of $M_2$, $(w'_1, \dots , w'_{d'})$ is a safe tuple with set of pivots $a_1, \dots , a_{d'}$. Suppose $a_j \in \text{BLOCK}(b_j)$. The blocks $b_1,  \dots , b_{d'}$ are distinct and all of them lie in $[m] \setminus T$. Let $\text{PIVOTS}= \{b_1, b_2, \cdots , b_{d'}\}$ and $\text{FREE}= [m] \setminus (T \cup \text{PIVOTS}).$
We construct $\gamma$ in two steps. In the first step we construct a $\tilde{\beta} \in \F^{mb}$ such that $\overrightarrow{g}(\tilde{\beta}) = \alpha$, but it is not necessarily the case that $M' \tilde{\beta} = c'$. In the second step we modify $\tilde{\beta}$ in the coordinates $a_1,  \dots , a_{d'}$ to get an assignment $\gamma \in \Space(M',c')$.

\begin{description}

    \item[Constructing $\tilde{\beta}$] $ $
        \begin{itemize}
        \item For each $i \in  T= \Cl(A)$, $\tilde{\beta}$ agrees with $\beta$ on $\text{BLOCK}(i)$, i.e., $(\tilde{\beta})^i = \beta^i.$
        \item For each $i \in \text{FREE}$, choose an arbitrary preimage $ u_i \in g^{-1}(\alpha_i)$ and set $(\tilde{\beta})^i = u_i.$
        \item For each $i = b_j \in \text{PIVOT}$: Suppose the pivot $a_j$ is the $\ell$-th coordinate of $\text{BLOCK}(j)$. Pick $u_i \in g^{-1}(\alpha_i)$ to be a stifling assignment for the $\ell$-th coordinate, i.e., $g(u_i) = g(u_i^{(l)})= \alpha_i$ (where $s^{(l)}$ denotes $s$ with $l$-th coordinate flipped). Set $(\tilde{\beta})^i= u_i.$ 
        \end{itemize}
    \item [Constructing $\gamma$:] We modify $\tilde{\beta}$ in the coordinates $a_1,  \dots , a_{d'}$ to get an assignment $\gamma$ in $\Space(M',c')$ as follows. For $1 \leq j \leq d$, let $f_j = \langle w'_j, \tilde{\beta} \rangle + (c')_j$. Let $\gamma \in \F^{mb}$ be the following assignment:
    \begin{align*}
        \gamma_i &= \begin{cases}
            (\tilde{\beta})_i & \text{ if }i \not \in \{a_1,  \dots,  a_{d'} \}, \\
            (\tilde{\beta})_i + f_j & \text{ if } i=a_j.
        \end{cases}
    \end{align*}
    
\end{description}

\begin{claim}
\label{claim:gamma}
    $\overrightarrow{g}(\gamma) = \alpha$ and $\gamma \in \Space(M',c')$.
\end{claim}

\begin{proof}
    We show both points separately. 
    \begin{description}
    \item\textbf{Showing $\overrightarrow{g}(\gamma)=\alpha$:} We argue that $g(\gamma^i)=\alpha_i$ for all $i \in [m]$. 
    \begin{description}
        \item[Case 1,] $i \in T$: We have set $(\tilde{\beta})^i=\beta_i$. Note that $\gamma$ differs from $\tilde{\beta}$ only in coordinates $a_1, a_2, \cdots , a_{d'}$. All these coordinates lie outside $\text{BLOCK}(T)$. Thus, $g(\gamma^i)= g(\beta^i) = \alpha_i$.
        \item[Case 2,] $i \in \text{FREE}$: We have set $(\tilde{\beta})^i = u_i$ where $u_i \in g^{-1} (\alpha_i)$. Again, note that $\gamma$ differs from $\tilde{\beta}$ only in the coordinates $a_1, a_2, \cdots , a_{d'}$, all of which lie outside $\text{BLOCK}(i)$. It follows that $g(\gamma^i)= \alpha_i$. 
        \item[Case 3,] $i \in \text{PIVOTS}$: Let $i= b_j$ and let $a_j \in \text{BLOCK}(b_j)$ be the corresponding pivot variable. Recall that each pivot variable lies in a distinct block. Let $a_j$ be the $\ell$-th coordinate of $\text{BLOCK}(b_j)$. We have set $(\tilde{\beta}^i) = u_i$ where $u_i \in g^{-1}(\alpha_i)$ is a stifling assignment for $\ell$ and $\alpha_i$. This means that $g((u_i)^{(\ell)}) = g(u_i) = \alpha_i$ ($s^{(\ell)}$ denotes $s$ with $\ell$-th coordinate flipped). Notice that $\gamma$ and $\tilde{\beta}$ agree everywhere on $\text{BLOCK}(b_j)$ except possibly $a_j$. This implies $g(\gamma^i)= \alpha_i$. 
    \end{description}

    \item{\textbf{Showing $\gamma \in \Space(M',c')$:} } Note that all equations corresponding to rows in $M_1$ are satisfied by $\gamma$ since they are satisfied by $\beta$ and hence by $\tilde{\beta}$ too. That the equations corresponding to $M_2$ are satisfied by $\gamma$ follows from the row echelon structure of $M_2$, i.e., the fact that after an appropriate permutation of the columns, $M'$ looks as follows:
\begin{align*}
\begin{array}{cc|c|c||c}
\multirow{2}*{M' =} & B_1 & 0   & 0 & = M_1 \\
& B_2   & I_{d'} & B_3 & = M_2 \\
\cline{2-5}
& \text{Closure $T$} & \text{Pivots of $M_2$} & \multicolumn{2}{c}{~} \\
\end{array}
\end{align*}      
\end{description}
        
\end{proof}

Since $\Space(M',c')= \Space(M,c)= A$, Lemma \ref{lem:StiflingGadget} follows immediately from Claim \ref{claim:gamma}.    
\end{proof}

\section{Proof Outline}
\label{sec:Outline}

In this section, we provide an outline of the proof of our main result, Theorem~\ref{thm:MainResult}.
The proof consists of two parts.
The first part shows that the formula $\Stone(G, \rho) \circ g$ has a polynomial length resolution proof for any directed acyclic graph $G$ on $N$ vertices and out-degree 2, any obfuscation map $\rho: [N]^3 \to \V$, and any gadget $g: \cube^b \to \cube$, where $b$ is logarithmic in $N$ (recall that the number of variables $m$ of the formula $\Stone(G, \rho)$ is $2N^2$).
This part of the proof is an adaptation of an analogous proof for the stone formula given by Alekhnovich et al.~\cite{AJPU02}.

The second part establishes that there is a graph $G$ and an obfuscation map $\rho: [N]^3 \to \V$ such that any bottom-regular ResLin proof of $\Stone(G,\rho) \circ \IP$ has exponential length in $m$, where $\IP$ is the inner product function on $b = \Theta(\log m)$ bits.
The proof of this part is involved and non-trivial. We outline the main steps in the figure below, immediately followed by a high-level description of each step depicted. 

\paragraph{Outline of the Lower Bound Proof.}
Our argument is an adaptation of the method presented in Efremenko et al \cite{EGI23} with addition of some new ingredients. See Figure~\ref{fig:ProofOutline}, for depicting the method. Let $\cP(\beta,t)$ be the node that $\cP$ arrives at after making $t$ linear queries on $\beta$.

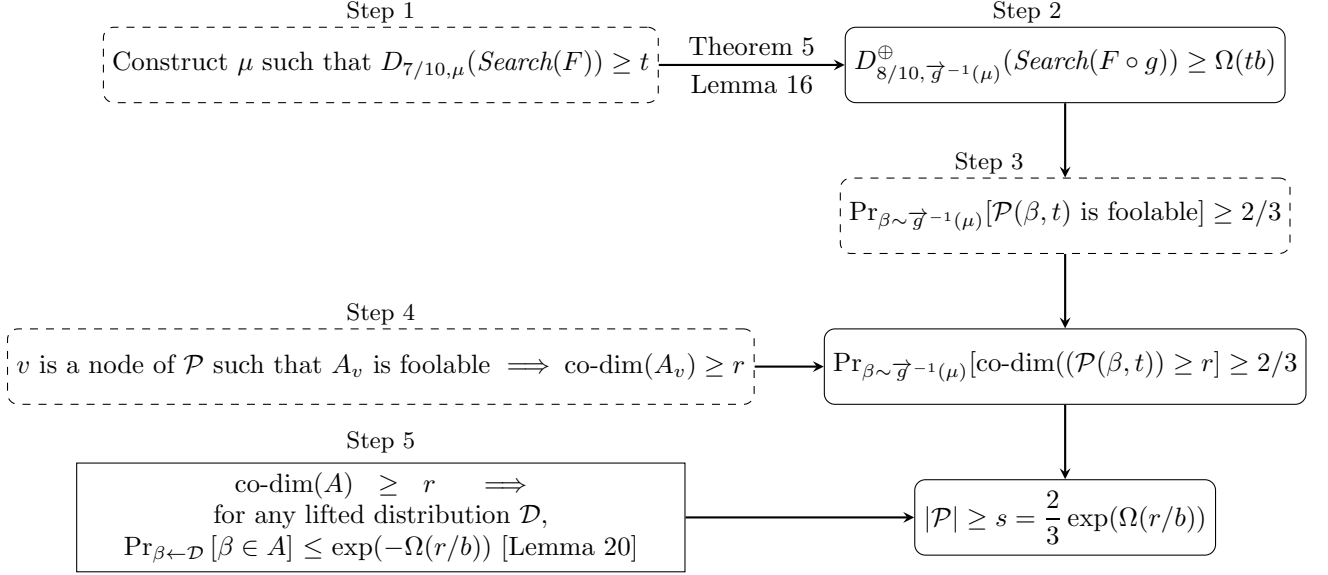
\begin{figure}[ht]
    \centering
\begin{tikzpicture}
    \node (start) [startstopdashed] {   
       Construct $\mu$ such that $D_{7/10, \mu} (\Search(F)) \geq t$};
    \node (captionstart) [abcd, above of=start, yshift= -0.3cm] {
    \small{Step \ref{LBstep:Tree}}
    };
    \node (lift) [startstop, right of=start, xshift= 8cm] {$D_{8/10, \overrightarrow{g}^{-1}(\mu)}^{\oplus} (\Search (F \circ g)) \geq \Omega (tb)$ };
    \node (captionlift) [abcd, above of= lift, xshift= -0.5cm, yshift= -0.3cm]{
    \small{Step \ref{LBStep:RandAssignment}}
    };
    \node (foolable) [startstopdashed, below of=lift, yshift=-1cm] {$ \Pr_{\beta \sim \overrightarrow{g}^{-1}(\mu)} [\mathcal{P} (\beta, t) \text{ is foolable}] \geq 2/3$};
    \node (captionfoolable) [abcd, above of= foolable, yshift= -0.3cm, xshift= -1cm]{
    \small{Step \ref{LBstep:foolable-frequent}}
    };
    
    \node (largerank) [startstopdashed, below of=start, yshift= -3cm] {$v$ is a node of $\mathcal{P}$ such that $A_v$ is foolable
    $\implies \text{co-dim}(A_v) \geq r$};
    \node (reachlargerank) [startstop, below of= foolable, yshift= -1 cm]{ 
    $ \Pr_{\beta \sim \overrightarrow{g}^{-1}(\mu)} [\text{co-dim} ((\mathcal{P}(\beta, t)) \geq r] \geq 2/3$
    };
      \node (captionlargerank) [abcd, above of=largerank, yshift= -0.3cm]{
    \small{Step \ref{LBstep:LargeRank}}
    };

    \node (smallmass) [process, below of= largerank, yshift= -1 cm]{
    $ 
    \text{co-dim}(A) \geq r \implies \text{for any lifted distribution } \mathcal{D},$
    $ \Pr_{\beta \leftarrow \mathcal{D}} \left[ \beta \in A \right] \leq \exp (-\Omega(r/b))$
    [Lemma \ref{lemma:rank-fooling}]
    }; 

    \node (captionsmallmass) [abcd, above of= smallmass]{
    \small{Step \ref{LBStep:PseudoRandom}}
    };
    \node (lowerbound) [startstop, below of= reachlargerank, yshift= -1 cm]{
    $ |\mathcal{P}| \geq s= \dfrac{2}{3} \exp (\Omega(r/b))$
    };
    \node (captionreachlargerank) [abcd, above of=reachlargerank, yshift= -0.3cm, xshift=-1cm]{
    \small{ }
    };
    
    \draw [arrow] (start) -- node[anchor=south] {Theorem \ref{thm:CFKMP-simulation}} node[anchor=north] {Lemma \ref{lemma:PDT-average}}(lift);
    \draw [arrow] (lift) --  (foolable);
    \draw [arrow] (foolable) -- (reachlargerank);
    \draw [arrow] (largerank) -- (reachlargerank);
    \draw [arrow] (reachlargerank) -- (lowerbound);
    \draw [arrow] (smallmass) -- (lowerbound);
\end{tikzpicture}

    \caption{Outline of the proof. The solid boxes refer to parts that are quite general and not specific to a formula, while the dashed boxes contain modules that are more specific to $\SP_n \circ \IP$ and similar formulas.}
    \label{fig:ProofOutline}
\end{figure}

\paragraph{Some Details:}
Let $\cP$ be a bottom-read-once branching program computing $\Search(\Stone(G,\rho) \circ \IP)$ corresponding to a bottom-regular ResLin proof of $\Stone(G,\rho) \circ \IP$ where $G$ is a pyramid graph of $n$ levels and $\rho$ is a carefully chosen obfuscation map.
The proof consists of several steps.
\begin{enumerate}
    \item \label{LBstep:Tree} We design a distribution $\mu$ over the assignment of variables of the base formula $F$ over $m$ variables, typically supported over \emph{critical assignments}, i.e. those which result in the falsification of exactly one clause. This module requires one to show that $\Search(F)$ is average-case hard for deterministic decision trees of small height wrt $\mu$. In particular, our Lemma~\ref{lemma:DT-error} proves that the problem $\Search(\Stone(G, \rho))$ is average-case hard for deterministic decision trees of height at most $O(n^{1/3})$. As the $\mu$ exhibited is formula specific, the box corresponding to this module is dashed.
    \item \label{LBStep:RandAssignment} In this step, we prove that the search problem associated with the lifted formula $F \circ g$ remains average case hard for \emph{parity decision trees} wrt a \emph{lifted distribution} as long as the gadget $g$ has small rectangular discrepancy. More precisely, let $\overrightarrow{g}^{-1}(\mu)$ denote the lifted distribution generated by the following sampling: sample an input $z \in \{0,1\}^m$ according to $\mu$. Then, sample at random an input $\beta \in \{0,1\}^{mb}$, conditioned on $\overrightarrow{g}(\beta) =z$. Using Theorem~\ref{thm:CFKMP-simulation}, implicit in the proof of the main result of Chattopadhyay, Filmus, Koroth, Meir and Pitassi \cite{CFKMP21}, we conclude that $\Search(F\circ g)$ is average-case hard for deterministic parity decision trees of small height, under the lifted distribution $\overrightarrow{g}^{-1}(\mu)$. This step is generic and works for any gadget of size $c\cdot \log(m)$, that has sufficiently small rectangular discrepancy under the uniform distribution over $\{0,1\}^b$. The gadget we use here is $\IP$.
    \item \label{LBstep:foolable-frequent} We then want to define a notion of progress the branching program $\mathcal{P}$ has made on arriving at a node $v$. To do so, consider the affine space $A_v$ that labels the node. $A_v$ may have nearly fixed/exposed the values of some of the blocks of input. These dangerous blocks are precisely $\Cl(A_v)$ as defined in Section~\ref{sec:linear}. They form the minimum obstruction set. Intuitively, the danger is $\mathcal{P}$ may have nearly found out a falsified clause of $F \circ g$ on reaching $v$ if that clause was made up entirely of variables from blocks in $\Cl(A_v)$. However, in this step we observe that the average-case hardness of the Search problem for PDTs proved in the previous step precludes this from happening with appreciable probability, when the input is sampled according to the lifted distribution $\overrightarrow{g}^{-1}(\mu)$. To formalize this idea, we need to concretely say when $A_v$ is (not) dangerous. So far, we have not been able to lay out a general notion of danger, but notions specific to individual formulas have been defined. For $\Stone(G,\rho)$, this notion is captured by Definition~\ref{defn:foolable-space} of \emph{foolable} spaces, provided in Section~\ref{sec:foolable-nodes-frequent}. Theorem~\ref{thm:node-foolability} shows that w.h.p., $\mathcal{P}$ reaches a foolable space on walking for $n^{1/3}$ steps, querying an input sampled according to $\overrightarrow{g}^{-1}(\mu)$.
    
    \item \label{LBstep:LargeRank} In this step, we show that when the affine space $A_v$ is not dangerous, i.e. it is foolable or consistent, the appropriate notion depending on the formula at hand, $A_v$ has large co-dimension. All steps until now held for general branching programs (or equivalently proof DAGs). This step is the only one where the bottom-read-once property is exploited. For $\Stone(G,\rho)$, this is achieved in Section~\ref{sec:foolability-to-rank}, at the end, by Lemma~\ref{lemma:large-rank-stone}. 
    \item \label{LBStep:PseudoRandom} In this step, we prove a general result about lifted distributions. For any affine space $A$ of $\codim(A) = r$ and any distribution $\mu$ on $\{0,1\}^m$, we prove that $\beta$ sampled by $\overrightarrow{g}^{-1}(\mu)$ is in $A$ only with probability $2^{-\Omega(r/b)}$, as long as the gadget $g$ is balanced and stifling. In other words, lifted distributions, even though their support is quite sparse in the ambient space, are pseudo-random for the rank measure. This property, though simple to prove, turns out to be extremely useful, especially for formulas like the stone formulas that are barely hard. 

\end{enumerate}
At this stage we are ready to put together the above steps in the following way.
Let $R$ be a set of nodes $w$ of $\cP$ such that there is a path from the root of $\cP$ to $w$ of length $t$, and $\codim(A_w) \geq t$.
Setting $t := n^{1/3}$ we have the following.
%
\begin{align*}
\frac{7}{10} &\leq \Pr_{\beta \sim \overrightarrow{\IP}^{-1}(\mu)}\bigl[\codim(A_v) \geq t \text{ for } v = \cP(\beta, t)\bigr] \tag{by Step~\ref{LBstep:foolable-frequent} and~\ref{LBstep:LargeRank}} \\
& \leq \sum_{w \in R} \Pr_{\beta \sim \overrightarrow{\IP}^{-1}(\mu)}\bigl[\cP(\beta, t) = w\bigr] \tag{by union bound} \\
& \leq |R|\cdot 2^{\Omega(-t/b)} \tag{by Step~\ref{LBStep:PseudoRandom}}
\end{align*}
By rearranging, we get the lower bound $|R| \geq 2^{\Omega(n^{1/3}/\log n})$. Recall that the number of variables of $\Stone(G,\rho) \circ \IP$ is $M= \Theta (n^4 \log(n))$. In terms of $M$, the lower bound is $2^{\Omega(M^{1/12}/\log^{13/12} M})= 2^{M^{\Omega(1)}}$.


\section{Upper Bound}
\label{sec:UpperBound}
In this section, we show the upper bound part of Theorem~\ref{thm:MainResult}.
\begin{theorem}
\label{thm:UpperBound}
    Let $G = (V,E)$ be an directed acyclic graph with $N$ vertices such that there is exactly one root $r$ (vertex with indegree 0), and each non-sink vertex has outdegree exactly 2.     
    Let $\rho: [N]^3 \to \V$ be any obfuscation map, and $g: \cube^b \to \cube$ be a Boolean function for $b \leq O(\log N)$.
    Then, the formula $\Stone(G, \rho) \circ g$ admits a resolution refutation of length polynomial in $N$.
\end{theorem}
The proof of Theorem~\ref{thm:UpperBound} is an adaptation of the proof given by Alekhnevich et al.~\cite{AJPU02} for lifted formulas.
We remark that Alekhnevich et al.~\cite{AJPU02} presented a resolution refutation for the stone formulas of constant width.
This allow us to adapt the refutation for the lifted formula.
For the rest of the section, we fix a graph $G$, an obfuscation map $\rho$, and a gadget $g$ satisfying the assumptions of Theorem~\ref{thm:UpperBound}.
First, we prove several auxiliary lemmas about the formula $\Stone(G, \rho) \circ g$.

\begin{lemma}
\label{lem:ResolveP}
    Let $C$ be a clause with width $w$. Suppose we have derived the clauses $C \lor \neg P_{v,j}$ for a fixed $v \in V$ and all $1 \leq j \leq N$. Then, we can derive $C$ in $N$ steps in width $\leq w+2$.
\end{lemma}
\begin{proof}
    We derive the clause $C$ in $N$ steps. We will subsequently derive $C \lor \neg Z_{v,j+_1}$ from $C \lor \neg Z_{v,j}$
    \begin{description}
        \item[Base step:] Deriving $C \lor \neg Z_{v,1}$.
         
\begin{forest}
  for tree={
    grow'=90,
    parent anchor=north,
    math content,
    before typesetting nodes={
      if level=0{}{
        if content={}{
          shape=coordinate
        }{
          content/.wrap value={\{#1\}},
        },
      },
    }
  }
  [C \lor \neg Z_{v,1}[
  [C \lor \neg P_{v,1}][P_{v,1} \lor \neg Z_{v,1}]
   ]
  ]
\end{forest}   

\item[Step $j$:] For $j \in [N-2]$, deriving $C \lor \neg Z_{v,j+1}$ from $C \lor \neg Z_{v,j}$.
        
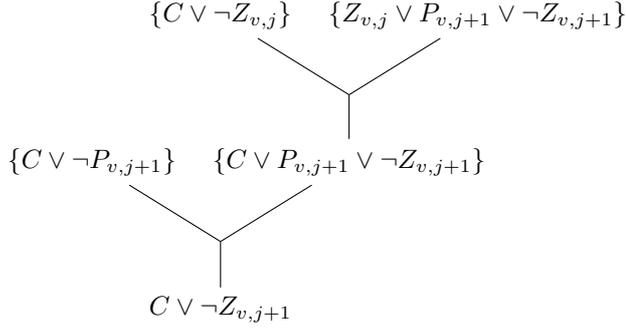
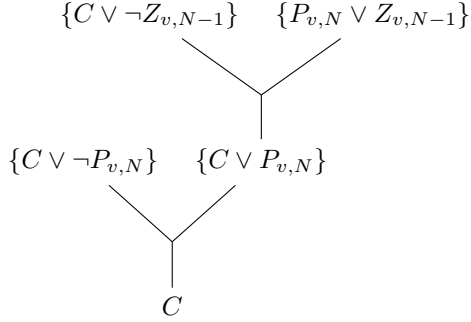
\begin{forest}
  for tree={
    grow'=90,
    parent anchor=north,
    math content,
    before typesetting nodes={
      if level=0{}{
        if content={}{
          shape=coordinate
        }{
          content/.wrap value={\{#1\}},
        },
      },
    }
  }
  [C \lor \neg Z_{v,{j+1}}[
  [C \lor \neg P_{v,j+1}][C \lor P_{v,j+1} \lor \neg Z_{v, j+1} [[C \lor \neg Z_{v,j}] [Z_{v,j} \lor P_{v,j+1} \lor \neg Z_{v,j+1}] ]]
   ]
  ]
\end{forest}  

\item[Final step:] Deriving $C$.

\begin{forest}
  for tree={
    grow'=90,
    parent anchor=north,
    math content,
    before typesetting nodes={
      if level=0{}{
        if content={}{
          shape=coordinate
        }{
          content/.wrap value={\{#1\}},
        },
      },
    }
  }
  [C [ [C \lor \neg P_{v,N}] [ C \lor P_{v,N} [ [C \lor \neg Z_{v,N-1}] [P_{v,N} \lor Z_{v,N-1}]]  ] ]]
\end{forest}  
\end{description}
    
\end{proof}
For a vertex $v$, we define the set of clauses $ S(v) =  \{ \neg P_{v,j} \lor  R_j | 1 \leq j \leq N \}$.

\begin{lemma}
\label{lem:RedStonesResolving}
    Let $v$ be a vertex in $G$ with children $v_0, v_1$. We can derive $S(v)$ from $S(v_0)$, and $S(v_1)$ in constant width and length $O(N^3)$.
\end{lemma}
\begin{proof}
We derive $S(v)$ in several steps.

    \begin{enumerate}
        \item For every $j, j_0, j_1 \in [N]$, we perform the following sequence of operations: 
        
 
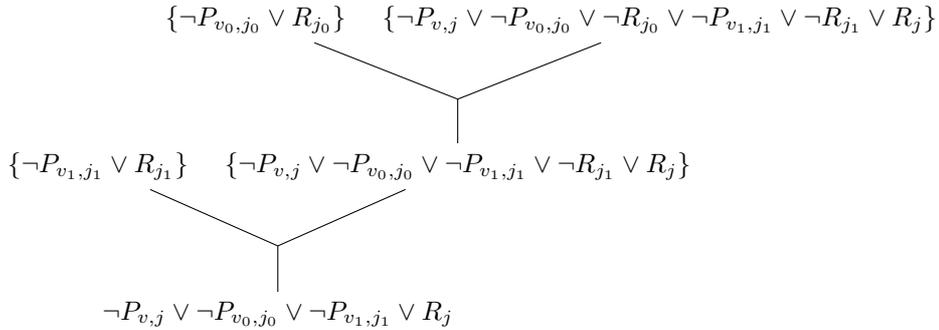
\begin{forest}
  for tree={
    grow'=90,
    parent anchor=north,
    math content,
    before typesetting nodes={
      if level=0{}{
        if content={}{
          shape=coordinate
        }{
          content/.wrap value={\{#1\}},
        },
      },
    }
  }
  [{\neg P_{v,j} \lor \neg P_{v_0, j_0} \lor \neg P_{v_1, j_1} \lor R_j} [ [\neg P_{v_1, j_1} \lor R_{j_1}] [\neg P_{v,j} \lor \neg P_{v_0, j_0} \lor \neg P_{v_1, j_1} \lor \neg R_{j_1} \lor R_j  [ [\neg P_{v_0, j_0} \lor R_{j_0}] [\neg P_{v,j} \lor \neg P_{v_0, j_0} \lor \neg R_{j_0} \lor \neg P_{v_1, j_1} \lor \neg R_{j_1} \lor R_j]]]]]
\end{forest}            
    \item For each fixed $j_0, j$, we apply Lemma~\ref{lem:ResolveP} to the clause $C \coloneqq \neg P_{v,j} \lor \neg P_{v_0, j_0} \lor R_j$ and we derive $\neg P_{v,j} \lor \neg P_{v_0,j_0} \lor R_j$.
    \item For each fixed $j$, we apply Lemma~\ref{lem:ResolveP} to the clause $C \coloneqq \neg P_{v,j} \lor R_j$ and we derive $\neg P_{v,j} \lor R_j $.
    
    \end{enumerate}    
\end{proof}
        
\begin{lemma}
\label{lem:StoneRefutation}
    The formula $\Stone(G, \rho)$ has a resolution refutation of width $O(1)$ and size polynomial in $N$.
\end{lemma}
\begin{proof}
The refutation proceeds in the following steps. 

\begin{description}

        \item[Elimination of the $\rho$'s]: For every induction clause $C$, we resolve the appended $\rho$-variable. 
        
 
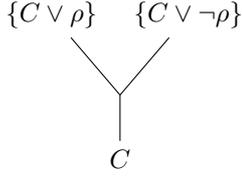
\begin{forest}
  for tree={
    grow'=90,
    parent anchor=north,
    math content,
    before typesetting nodes={
      if level=0{}{
        if content={}{
          shape=coordinate
        }{
          content/.wrap value={\{#1\}},
        },
      },
    }
  }
  [C [ [C \lor \rho] [C \lor \neg \rho]] ]
\end{forest}        

    \item[Derivation of $S(r)$:] For each sink $s$ of $G$, the clauses $S(s)$ are present in the axioms of $\Stone(G,\rho)$. By Lemma~\ref{lem:RedStonesResolving}, we subsequently derive the set $S(r)$ for the root $r$ of $G$.
    \item[Empty clause derivation:]  For each $1 \leq j \leq N,$ we derive $\neg P_{r,j}$.
    
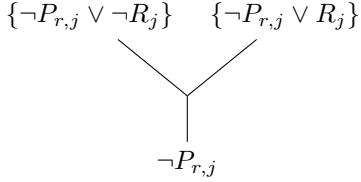
 \begin{forest}
  for tree={
    grow'=90,
    parent anchor=north,
    math content,
    before typesetting nodes={
      if level=0{}{
        if content={}{
          shape=coordinate
        }{
          content/.wrap value={\{#1\}},
        },
      },
    }
  }
  [\neg P_{r,j} [ [\neg P_{r,j} \lor \neg R_j] [\neg P_{r,j} \lor R_j]]]
\end{forest}      

    Now by applying Lemma~\ref{lem:ResolveP} for $C$ being an empty clause $\perp$, we derive $\perp$, that concludes the proof.
    \end{description}
\end{proof}

Now, from constant-width polynomial-length refutation of $\Stone(G, \rho)$ we derive a polynomial-length refutation of the lifted formula $\Stone(G,\rho) \circ g$.

\begin{lemma} \label{lemma:lifted-upper-bound}
    Let $g: \cube^b \to \cube$ be a Boolean function and $\Phi$ be a CNF unsatisfiable formula over $n$ variables containing only constant width clauses.
    Suppose $\Phi$ has a resolution refutation of length $\ell$ and constant width.
    Then, $\Phi \circ g$ contains clauses of width $O(b)$ and admits a resolution refutation of size $\ell \cdot 2^{O(b)}$.
\end{lemma}
\arkadev{We need to reference some basic result about resolution here. Put the above to preliminaries about resolution. }

\begin{proof}
    By construction, if $C$ is a clause of width $k$, then $|C\circ g| \leq 2^{bk}$. 
    If $k$ is constant, this is $2^{O(b)}$. 
    We show that, for every derivation step $(A \lor x), (B \lor \neg x) \rightarrow (A \lor B)$ in a proof for $\Phi$, we can derive all clauses of $(A \lor B) \circ g$ from the clauses of $(A \lor x) \circ g$ and $(B \lor \neg x) \circ g$ in polynomial size, assuming each of $A,B$ has constant width.
    This follows from the fact that $(A \lor x) \circ g$ and  $(B \lor \neg x) \circ g$  semantically imply $(A \lor B)\circ g$: an assignment $(x_{i,1}, \dots , x_{i,b})_{i \in [M]}$ satisfies formula $C \circ g$ if and only if the assignment $(g(x_{i,1}, \dots , x_{i,b}))_{i \in [n]}$ satisfies clause $C$.
    And since clauses $A \lor x, B \lor \neg x$ semantically imply $A \lor B$, it follows that the formulas $(A \lor x) \circ g$ and  $(B \lor \neg x) \circ g$ semantically imply the formula $(A \lor B) \circ g$.

    As both $A$ and $B$ are constant-width clauses, each of the formulae $(A \lor x) \circ g$ and $(B \lor \neg x) \circ g$  are defined on at most $O(b)$ variables. By Lemma~\ref{lemma:resolution-complete}, we can derive each clause in $(A \lor B) \circ g$ from $(A \lor x) \circ g$  and $(B \lor \neg x) \circ g$ in at most $2^{O(b)}$ resolution steps. 

    Using this fact, we can mimic the resolution refutation of $\Phi$.
    For each intermediate clause $C$ derived in the resolution refutation for $\Phi$, we can derive all clauses in $C \circ g$. 
    In the end, we derive $\perp \circ g = \{\perp\}$, i.e., the empty clause. 
    Assuming the width of the resolution refutation for $\Phi$ is bounded by some constant, the total length of our simulation is at most $\ell \cdot 2^{O(b)}$.
\end{proof}

Now, Theorem~\ref{thm:UpperBound} is a corollary of Lemma~\ref{lem:StoneRefutation}, and~\ref{lemma:lifted-upper-bound}.

\section{Lower Bound}
\label{sec:LowerBound}
In this subsection, we prove the lower bound part of Theorem~\ref{thm:MainResult} following the outline given in Section~\ref{sec:Outline}.

\begin{theorem}
\label{thm:LowerBound}
    There is an obfuscation map $\rho:[N]^3 \to \V$ such that any bottom-regular ResLin refutation of $\SP_{n,\rho} \circ \IP$ must have length at least $2^{\Omega(n^{1/3}/\log n)}$.
\end{theorem}

Recall that number of variables of $\SP_{n,\rho} \circ \IP$ is $\Theta(n^4 \log n)$.
Thus, the lower bound given by Theorem~\ref{thm:LowerBound} yields the lower bound claimed in Theorem~\ref{thm:MainResult}.
For the rest of this section, we fix $G = (V,E)$ to be the pyramid graph of $n$ levels, and $N = n(n+1)/2$ vertices.

\subsection{The Stone Formula is Average-Case Hard for Decision Trees}

We shall construct a distribution $\mu$ on $\{0,1\}^m$ such that for any obfuscation map $\rho: [N]^3 \to \V$, the search problem $\Search(\SP_{n,\rho})$ is hard on average w.r.t. $\mu$ for deterministic decision trees of sufficiently small height (around $n^{1/3}$). 

First, we fix an arbitrary bijection $f:[N] \to V$ between stones and vertices of the pyramid. All assignments in $\Supp(\mu)$ will place the stone $i$ on vertex $f(i)$.
The distribution $\mu$ samples the assignments as follows.
\begin{enumerate}
    \item Assign stone $i$ to vertex $f(i)$. Formally for each $v \in V$, $i \in [N]$, and $j \in [N - 1]$, we set:
        \begin{align*}
            & P_{v, i} =  \begin{cases}
            1 & \text{ if } f(i)=v\\
            0 & \text{ otherwise}
            \end{cases}
            \\
            &Z_{v,j} = \begin{cases}
                0 & \text{ if }j < f^{-1}(v) \\
                1 & \text{ otherwise}
            \end{cases}
        \end{align*}
    \item Sample $n - 2$ independent uniform bits $B_2, \dots , B_{n - 1} \in \{0,1\}$
    \item Let $X_1 =1$, and for $2 \leq j \leq n - 1$, let $X_j = X_{j - 1} + B_j$. Color the vertices $(j, X_j)$ blue for $1 \leq j \leq n - 1$ and other vertices red, i.e., for each stone $i \in [N]$, we set:
        \[
            R_i= \begin{cases}
                0 & \text{ if } j \leq n - 1 \text{ and } (j,X_j) = f(i)\\
                1 & \text{ otherwise}
            \end{cases}
        \]
\end{enumerate}

Let $\alpha \in \Supp(\mu)$.
The assignment $\alpha$ corresponds to the following stone placement.
It places a different stone on each vertex.
There is a path $P$ from the root $r = (1,1)$ to a vertex $v$ in the level $n - 1$ given by the random variables $X_1,\dots,X_{n-1}$, i.e. the vertices of the path are $\{(1,X_1),\dots,(n-1,X_{n - 1})\}$.
The stones on the vertices of $P$ are colored blue, all other stones are colored red. An example of such a coloring is shown in Figure \ref{fig:alpha}.

\begin{figure}[ht]
    \centering
    
\begin{tikzpicture}
\node (1001)[bluecircle]{};
\node (2001)[redcircle, below of = 1001]{};
\node (2002)[bluecircle, right of = 2001]{};
\node (3001)[redcircle, below of = 2001]{};
\node (3002)[bluecircle, right of = 3001]{};
\node (3003)[redcircle, right of = 3002]{};
\node (4001)[redcircle, below of = 3001]{};
\node (4002)[redcircle, right of = 4001]{};
\node (4003)[bluecircle, right of = 4002]{};
\node (4004)[redcircle, right of = 4003]{};
\node (5001)[redcircle, below of = 4001]{};
\node (5002)[redcircle, right of = 5001]{};
\node (5003)[redcircle, right of = 5002]{};
\node (5004)[bluecircle, right of = 5003]{};
\node (5005)[redcircle, right of = 5004]{};
\node (6001)[redcircle, below of = 5001]{};
\node (6002)[redcircle, right of = 6001]{};
\node (6003)[redcircle, right of = 6002]{};
\node (6004)[redcircle, right of = 6003]{};
\node (6005)[bluecircle, right of = 6004]{};
\node (6006)[redcircle, right of = 6005]{};
\node (7001)[redcircle, below of = 6001]{};
\node (7002)[redcircle, right of = 7001]{};
\node (7003)[redcircle, right of = 7002]{};
\node (7004)[redcircle, right of = 7003]{};
\node (7005)[redcircle, right of = 7004]{};
\node (7006)[redcircle, right of = 7005]{};
\node (7007)[redcircle, right of = 7006]{};
\draw [arrow](1001) -- (2002);
\draw [arrow](1001) -- (2001);
\draw [arrow](2001) -- (3002);
\draw [arrow](2001) -- (3001);
\draw [arrow](2002) -- (3003);
\draw [arrow](2002) -- (3002);
\draw [arrow](3001) -- (4002);
\draw [arrow](3001) -- (4001);
\draw [arrow](3002) -- (4003);
\draw [arrow](3002) -- (4002);
\draw [arrow](3003) -- (4004);
\draw [arrow](3003) -- (4003);
\draw [arrow](4001) -- (5002);
\draw [arrow](4001) -- (5001);
\draw [arrow](4002) -- (5003);
\draw [arrow](4002) -- (5002);
\draw [arrow](4003) -- (5004);
\draw [arrow](4003) -- (5003);
\draw [arrow](4004) -- (5005);
\draw [arrow](4004) -- (5004);
\draw [arrow](5001) -- (6002);
\draw [arrow](5001) -- (6001);
\draw [arrow](5002) -- (6003);
\draw [arrow](5002) -- (6002);
\draw [arrow](5003) -- (6004);
\draw [arrow](5003) -- (6003);
\draw [arrow](5004) -- (6005);
\draw [arrow](5004) -- (6004);
\draw [arrow](5005) -- (6006);
\draw [arrow](5005) -- (6005);
\draw [arrow](6001) -- (7002);
\draw [arrow](6001) -- (7001);
\draw [arrow](6002) -- (7003);
\draw [arrow](6002) -- (7002);
\draw [arrow](6003) -- (7004);
\draw [arrow](6003) -- (7003);
\draw [arrow](6004) -- (7005);
\draw [arrow](6004) -- (7004);
\draw [arrow](6005) -- (7006);
\draw [arrow](6005) -- (7005);
\draw [arrow](6006) -- (7007);

\end{tikzpicture} 
    \caption{An example of a pyramid graph with coloring giving by an assignment sampled by the hard distribution $\mu$. }
    \label{fig:alpha}
\end{figure}
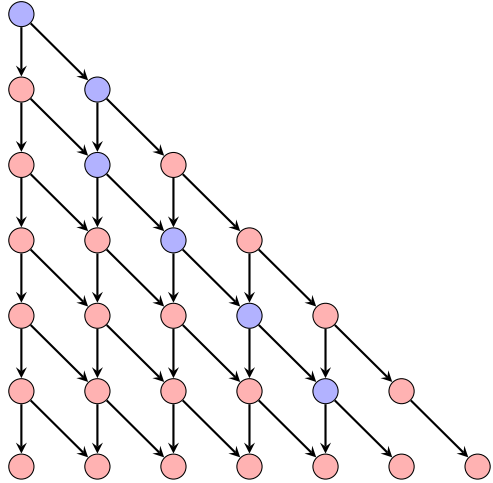

We call the path $P$ as the \emph{blue path induced by $\alpha$} and the vertex $v$ as the \emph{end} of $P$.
Note that the only clause falsified by the assignment $\alpha$ is one of the induction clauses for the vertex $v$ and its children $u$ and $w$.
Formally, for the stones $i = f^{-1}(u), j = f^{-1}(w)$, and $k = f^{-1}(v)$, the assignment $\alpha$ falsifies exactly one of the following two induction clauses (depending how $\alpha$ sets the value of $\rho(i,j,k)$):
\[
D_1(v) := \neg P_{u, i} \lor \neg R_{i} \lor \neg P_{w, j} \lor \neg R_{j} \lor \neg P_{v,k} \lor R_k \lor \rho(i, j, k)
\] 
\[
D_0(v) := \neg P_{u, i} \lor \neg R_{i} \lor \neg P_{w, j} \lor \neg R_{j} \lor \neg P_{v,k} \lor R_k \lor \neg \rho(i, j, k)
\]


Consider a random walk $Y_1,\dots,Y_k$ on a number line starting at $q \in \N$ distributed as follows: $Y_1 = q$, and for $i > 1$
\[
Y_i = 
\begin{cases}
    Y_{i-1} + 1 & \text{with probability $\frac{1}{2}$} \\
    Y_{i-1} & \text{with probability $\frac{1}{2}$} \\
\end{cases}
\]
Note that the random variables $X_1,\dots,X_{n-1}$ used in the construction of $\mu$ are distributed as $Y_1,\dots,Y_{n-1}$ for $Y_1 = 1$.

\begin{lemma}  \label{lemma:anti-concentration}
    There exists a constant $c_1 \geq 0$ such that for any $p \in \{q,\dots,q + k - 1\}$ and $t \in \{2,\dots,k\}$, we have $\Pr [Y_t = p] \leq c_1/\sqrt{t}$.
\end{lemma}
\begin{proof}
   Note that $Y_t = q + \sum_{i = 2}^t B_i$, where each $B_i$ is an independent uniform random bit. Now, $Y_t = p = p' + q$ for $p' \in \{0,\dots,k-1\}$ if and only if $\sum_{i = 2}^t B_i = p'$.
   \[
   \Pr [Y_t=p] = \Pr[\sum^t_{i=2} B_i = p'] =\binom{t-1}{p'}\cdot 2^{-t + 1} \leq \binom{t - 1}{\lfloor \frac{t-1}{2} \rfloor}\cdot 2^{-t + 1} \leq \frac{c_1}{\sqrt{t}}
   \]
   For an appropriate constant $c_1 > 0$, the last inequality is implied by Stirling's formula.
\end{proof}

For a set $S \subseteq [k] \times \N$, we say the random walk $W$ avoids $S$ if for all $(i,j) \in S$ it holds that $Y_i \neq j$.
\begin{lemma}  \label{lemma:forbidden}
    Let $c_2 \geq 1$ be a constant and $S \subseteq [k] \times \mathbb{N}$ be a set of forbidden points with $|S| \leq t$.
    Suppose there exists an interval $I = [L,R] \subseteq [k]$ with $|I| \geq c \cdot t^2$ for sufficiently large constant $c$ depending on $c_2$, such that no point of $S$ has the first coordinate in $I$, i.e., for all $(i,j) \in S: i<L$ or $i>R$.
    If the random walk $W = Y_1,\dots,Y_k$ avoids $S$ with non-zero probability, then for any $z \in \{q,\dots,q + k - 1\}$ it holds that 
    \[
    \Pr \bigl[Y_k = z \mid W \text{ avoids } S\bigr] \leq \frac{1}{c_2 t}.
    \] 
\end{lemma}

\begin{proof}
    We partition $S$ into two subsets $S_1$ and $S_2$ of points before and after the interval $I$,   $S_1 = \{(i,j) \in S \mid i < L \}$,  and $S_2 = \{(i,j) \in S \mid i > R \}$.
    Note that by the assumption, we have $S = S_1 \dot \cup S_2$.
    
    We show that for all $p$ such that $\Pr \bigl[Y_L = p \mid W \text{ avoids } S \bigr] > 0$, we have $\Pr \bigl[Y_k=z \mid Y_L=p, W \text{ avoids } S\bigr] \leq  1/ c_2 t $. 
    We set $c := 4c^2_1c^2_2$ where $c_1$ is the constant given by Lemma~\ref{lemma:anti-concentration}.
    We have by Lemma~\ref{lemma:anti-concentration} that
    \begin{equation}
        \label{eq:ReachingZ}
        \Pr \bigl[Y_k= z \mid Y_L= p, W \text{ avoids } S_1\bigr] = \Pr\bigl[Y_k= z \mid Y_L= p\bigr] \leq \frac{1}{2c_2t},
    \end{equation}
    since $z - L \geq |I| \geq c \cdot t^2$.
    Again, for all $(i,j) \in S_2$ (i.e., $i>R$) we have by Lemma~\ref{lemma:anti-concentration} that
    \[
    \Pr \bigl[ Y_i=j \mid Y_L=p, W \text{ avoids } S_1 \bigr] \leq \frac{1}{2c_2t} \leq \frac{1}{2t}.
    \]
    By union bound,
    \begin{equation}
    \label{eq:AvoidingProbability}
    \Pr \bigl[\exists (i,j) \in S_2 \text{ such that } Y_i=j  \mid Y_L = p, W \text{ avoids } S_1 \bigr] \leq \frac{1}{2}.    
    \end{equation}
    Therefore,
    \begin{align*}
        \Pr[Y_k = z \mid Y_L = p, W \text{ avoids } S ] &= \frac{\Pr\bigl[Y_k = z, W \text{ avoids } S_2 \mid Y_L = p, W \text{ avoids } S_1\bigr]}{\Pr\bigl[W \text{ avoids } S_2 \mid Y_L = p, W \text{ avoids } S_1\bigr]} \\
    &\leq 2 \cdot \Pr\bigl[Y_k = z \mid Y_L = p, W \text{ avoids } S_1\bigr] \tag{by~(\ref{eq:AvoidingProbability})} \\
    & \leq \frac{1}{c_2t}  \tag{by~(\ref{eq:ReachingZ})}
    \end{align*}
    Now, we are ready to finish the proof.
    \begin{align*}
        \Pr \bigl[Y_k= z | W \text{ avoids } S\bigr] &= \sum_p \Pr\bigl[Y_L = p \mid W \text{ avoids } S] \cdot \Pr\bigl[Y_k = z \mid Y_L = p, W \text{ avoids } S \bigr] \\
        &\leq  \sum_p \Pr\bigl[Y_L = p \mid W \text{ avoids } S] \cdot\frac{1}{c_2 t}  =  \frac{1}{c_2 t} 
    \end{align*}

    \end{proof}

Now we show that when inputs are sampled according to $\mu$, any deterministic decision tree for $\Search(\SP_{n,\rho})$ with small height makes an error with high probability. Note that for each $v \in V, i \in [N], j \in [N-1]$, the assignment to the variables $P_{v,i}$ and $Z_{v,j}$ are fixed by any assignment in $\Supp(\mu)$ (in other words, for each vertex the stone placed on it is fixed). 
Thus, we can assume WLOG the decision tree only queries the variables $R_j$. 
We say a decision tree queries the color of a vertex $v$ of $G$ if it queries the variable $R_{f^{-1}(v)}$. (Recall that the stone $f^{-1}(v)$ is placed on the vertex $v$ by assignments in $\Supp(\mu)$.)

Note that there is a simple, even non-adaptive, decision tree of height $O(\sqrt{n})$ that makes few errors. 
It simply queries the colours of $O(\sqrt{n})$ nodes of the pyramid graph, centered around the $(n-1)/2$-th node at level $n-1$. 
With very high probability, there is a blue node among the queried ones which uniquely identifies the falsified induction clause. 
Nevertheless, we will show that all decision trees of height at most $n^{1/3}$, will make errors with large probability to identify a falsified clause.

Consider a decision tree $\cT$ for $\Search(\SP_{n,\rho})$.
We transform $\cT$ into a decision tree $\cT'$ in a canonical form:
\begin{itemize}
    \item Initially, $\cT'$ always queries the color of the root $r$ of $G$.
    \item Suppose $\cT$ outputs an induction clause $D_0(v)$ or $D_1(v)$ for a vertex $v$ of $G$. Then, $\cT'$ queries the color of the vertex $v$ first. If the color of $v$ is red (i.e., $R_{f^{-1}(v)} = 1$), then $\cT'$ outputs an error symbol. Otherwise it outputs the same induction clause that $\cT$ outputs.
    \item If $\cT$ outputs any other clause, then $\cT'$ outputs an error symbol.
\end{itemize}

We remark this modification increases the height of the tree by at most two. 
Given any assignment in $\Supp(\mu)$, a decision tree $\cT'$ in a canonical form can output either an induction clause from $\{D_0(v), D_1(v) \mid v \in V(G)\}$ or an error symbol. The probability of $\cT$ making an error is precisely the probability of reaching a leaf node of $\cT'$ labeled with an error symbol.

Note that for each cube $C \subseteq \{0,1\}^m$ there is a corresponding partial assignment $\alpha_C \in \{0,1,*\}^m$ such that the cube $C$ is exactly the set of total extensions of $\alpha_C$, i.e., $C = \{\alpha \in \cube^m \mid \alpha \text{ extends } \alpha_C \}$.
We say a cube $C \subseteq \cube^{m}$ fixes a vertex $v \in V(G)$ to red (or blue) if the corresponding partial assignment $\alpha_C$ assigns a value 1 (or 0) to the variable $R_{f^{-1}(v)}$.

Now, fix a decision tree $\cT$ for $\Search(\SP_{n, \rho})$ in a canonical form and let $h := \gamma \cdot n^{1/3}$ be the height of $\cT$, where $\gamma > 0$ is sufficiently small constant.

\begin{definition}
\label{defn:useful_node}
    We say a cube $C \subseteq \{0,1\}^m$ is \emph{useful} if there exist $1 \leq L_1 \leq L_2 \leq L_3 \leq L_4 \leq N$ such that:
    \begin{enumerate}
        \item The cube $C$ fixes the some vertex in level $L_1$ to blue and some vertex in level $L_4$ to blue.
        \item For all $L_2 \leq \ell \leq L_3$, the cube $C$ does not fix the color of any vertex in level $\ell$.
        \item $L_3 - L_2 \geq \dfrac{n}{2h}$
    \end{enumerate}
    A node $p$ of $\cT$ is called \emph{useful} if the cube $C_p$ associated with it is useful.
\end{definition}
Clearly, the root of $\cT$ is not useful.

\begin{lemma}  \label{lemma:leaf-bad}
    Let $\alpha \in \Supp(\mu)$ be an assignment on which $\cT$ reaches the leaf $p$. If $p$ does not output an error symbol, then $p$ is useful.
\end{lemma}
\begin{proof}
    Let $v$ be the endpoint of the blue path induced by $\alpha$. Then, $p$ outputs one of the induction clauses $D_0(v)$ or $D_1(v)$ for $v \in V(G)$. Since $\cT$ is in the canonical form, the cube $C_p$ fixes the vertex $v$ and the root $r$ of $G$ to blue.
    Recall that the vertex $v$ is in level $n-1$. 
    Let $1 = \ell_1 < \cdots < \ell_d=n-1$ be the levels where $C_p$ fixes some vertices.
    Since $d \leq h$, there must exist an $i$ such that $\ell_{i+1} - \ell_i \geq \dfrac{n-1}{h+1} > \dfrac{n}{2h}. $.
    There is no fixed vertex on levels $\ell_i + 1,\dots,\ell_{i+1} - 1$. 
    
    We take largest the $\ell_1$ such that $\ell_1 \leq \ell_i$ and $C_p$ fixes a vertex on $\ell_1$ to blue.
    Similarly, we take the smallest $\ell_2$ such that $\ell_2 \geq \ell_{i + 1}$ and $C_p$ fixes a vertex on $\ell_2$ to blue.
    The cube $C_p$ satisfies the conditions of being a useful node in Definition \ref{defn:useful_node} by taking $(L_1, L_2, L_3, L_4)= (\ell_1, \ell_i+1, \ell_{i+1}-1, \ell_2)$.
\end{proof}

\begin{lemma}  \label{lemma:reaching-bad-node}
For any $\varepsilon > 0$ there exists $\gamma > 0$ such that 
    \[
    \Pr_{\alpha \sim \mu} [\text{The computation path of $\cT$ on } \alpha \text{ reaches a useful node}] \leq \varepsilon.
    \] 
\end{lemma}

\begin{proof}
    Let $\mathcal{T}(\alpha, k)$ denote the node of $\mathcal{T}$ reached by $\alpha$ after $k$ queries.
    For each $1 \leq k \leq h,$ we upper bound the probability that the computation path of $\alpha$ reaches a useful node for the first time at step $k$. 
    Then, we shall use union bound on $k$.
    Formally, we bound the probability as follows.
    \begin{align*}
        \Pr_{\alpha \sim \mu} & [\text{computation path of } \alpha \text{ reaches a useful node}] \\
        &= \Pr_{\alpha \sim \mu} \bigl[ \exists k \in [h]: \cT(\alpha,k) \text{ is useful and } \cT(\alpha,k-1) \text{ is not useful}\bigr] \\
        &\leq \sum^h_{k = 1} \Pr_{\alpha \sim \mu} \bigl[ \cT(\alpha,k) \text{ is useful and } \cT(\alpha,k-1) \text{ is not useful}\bigr] \\
        &\leq h \cdot \max_{k \in [h]} \Pr_{\alpha \sim \mu} \bigl[ \cT(\alpha,k) \text{ is useful} ~|~ \cT(\alpha,k-1) \text{ is not useful}\bigr]
    \end{align*}

    We bound the last probability for any $k \in [h]$.
    Let $p = \cT(\alpha,k-1)$.
    We assume the node $p$ is not useful.
    Let $(i,j) \in V(G)$ be the lowest vertex that is fixed by $C_p$ to blue. 
    Suppose that in the the next step $\cT$ queries a color of the vertex $(i', j')$. 
    If the next node has to be useful, the response to the query has to be blue. 
    Moreover, there have to be $n/2h$ consecutive layers $i < \ell',\dots,\ell' + \frac{n}{2h} - 1 < i'$ such that $C_p$ does not fix any vertex on those layers.
    The probability that the response to the query is blue is
    \[
    \Pr_{\alpha \sim \mu|C_p} \left[ \text{The blue path induced by  } \alpha \text{ visits }(i',j') \right].
    \] 
    
    Consider the random walk $X_1,\dots,X_{n-1}$ that determines the blue path $P$ induced by $\alpha \sim \mu$.
    Recall that the vertices of $P$ are $\{(1,X_1),\dots,(n-1,X_{n-1})\}$.
    The cube $C_p$ fixes colors of some vertices of $G$. Let $B$ and $R$ be the set of vertices whose colors are fixed by $C_p$ to blue and red respectively. 
    
    Conditioning on the cube $C_p$ restricts the random walk $X_1,\dots,X_{n-1}$ that it must visit the points in $B$ and must not visit the points in $R$.
    Formally, for any $(q,y) \in B$ it holds that $X_q = y$ and for any $(q',y') \in R$ it holds that $X_{q'} \neq y'$.
    We know there is at least one walk that avoids $R$ and visits $B$ (the walk corresponding to $P$).
    Moreover, we have $|R| \leq h$ and there is a large "gap" in $R$, i.e., for each $(i_1,j_1) \in R$ it holds that $i_1 < \ell'$ or $i_1 > \ell' + \frac{n}{2h} - 1$.
    Recall that we set $h = \gamma n^{1/3}$. Set $\gamma$ to a sufficiently small constant so that $\dfrac{n}{2h} \geq c h^2$, where $c$ is a sufficiently large constant for which we can apply Lemma \ref{lemma:forbidden} with $c_2= \varepsilon^{-1}$.
    By applying an appropriate time shift, we have by Lemma~\ref{lemma:forbidden} that
    \[
    \Pr_{\alpha \sim \mu|C_p} \bigl[  \text{The blue path induced by  }  \alpha \text{ visits }(i',j') \bigr] = \Pr_{\mu|C_p}[X_{i'} = j'] \leq \frac{1}{c_2 h}  \leq \frac{\varepsilon}{h}.
    \]
    Thus, we conclude that for all $k \in [h]$, 
    $$ \prob_{\mu \sim \alpha} \bigl[ \cT(\alpha,k) \text{ is useful } \mid \cT(\alpha,k-1) \text{ is not useful}\bigr] \leq \dfrac{\varepsilon}{h}. $$
    Therefore, the probability that the computation path ever reaches a useful node is at most $\varepsilon$.

\end{proof}

We end this subsection with by showing that the formula $\SP_{n,\rho}$ is average-case hard for decision trees.
\begin{lemma}   \label{lemma:DT-error}
    For any $\varepsilon > 0$, there exists $\gamma >0$ such that every deterministic decision tree of height at most $\gamma\cdot n^{1/3}$ for $\Search(\SP_{n,\rho})$ makes error with probability $\geq 1-\varepsilon$ w.r.t. the distribution $\mu$.
\end{lemma}
\begin{proof}
    If the decision tree answers correctly, by Lemma~\ref{lemma:leaf-bad} it must reach a useful node at some point. 
    By Lemma~\ref{lemma:reaching-bad-node}, the probability of this ever happening is at most $\varepsilon$ if $\gamma$ is small enough.
\end{proof}

\subsection{Lifting the Average-Case Hardness to Parity Decision Trees}

We lift the distribution $\mu$ to a distribution $\mu'$ of variables of $\SP_{n,\rho} \circ \IP$ as follows:
\begin{enumerate}
    \item Sample an assignment $\alpha$ according to $\mu$.
    \item Sample a uniformly random assignment from $\overrightarrow{\IP}^{-1}(\alpha)$.
\end{enumerate}
We remark that an assignment $\beta$ sampled by $\mu'$ falsifies exactly one clause of $\SP_{n,\rho} \circ \IP$, in particular one clause that arises by a lifting clause $C$ of $\SP_{n,\rho}$ where $C$ is the unique clause falsified by the assignment $\overrightarrow{\IP}(\beta)$.

In this section, we prove $\Search(\SP_{n,\rho} \circ \IP)$ is average-case hard for parity decision trees of small height under the lifted distribution. To do so, we shall use a result of Chattopadhyay et al.~\cite{CFKMP21}, that built upon the earlier work of G{\"o}{\"o}s, Pitassi and Watson \cite{GPW20}.

We will need to consider randomized decision trees that output Boolean strings in $\{0,1\}^t$, rather than $0/1$. For a given deterministic 2-party communication protocol $\Pi$, let $\Pi(x,y)$ denote the transcript generated by $\Pi$ on input $(x,y)$.

\begin{theorem}[Implicit in \cite{CFKMP21}]  \label{thm:CFKMP-simulation}
Assume $b \ge 50 \log(m)$. Let $\Pi$ be any deterministic 2-party communication protocol of cost $c$, where Alice and Bob each get inputs from $\{0,1\}^{mb}$. For any $z \in \{0,1\}^m$, let $(X_z,Y_z)$ denote the distribution on pairs obtained by sampling from $\overrightarrow{\IP}^{-1}(z)$ uniformly at random.  Then, there exists a randomized decision tree $\cT$ of cost $O(c/b)$ such that the following holds for every $z \in \{0,1\}^m$:
$$\TV \big(\cT(z),\Pi(X_z,Y_z) \big)\leq 1/10.$$
\end{theorem}

The above theorem says that a randomized decision tree is able to simulate by probing only a few bits of its input $z$, the transcript of a deterministic communication protocol when it is given a random input pair $X_z,Y_z$. 
Its relevance for us is due to the following simple observation.
\begin{observation}  \label{rem:pdt-to-protocol}
Every deterministic parity decision tree of height $h$ can be simulated exactly by a deterministic 2-party communication protocol of cost at most $2h$.
\end{observation}

Now we can lift our avarage-case hardness to parity decision trees.
\begin{lemma}  \label{lemma:PDT-average}
    There exists a constant $c > 0$ such that for every obfuscation map $\rho$ and every parity decision tree $\mathcal{T}$ of height at most $c\cdot n^{1/3}\log n$ purporting to solve $\Search\big(\SP_{n,\rho}\circ \IP\big)$, the following is true:
         $$ \Pr_{\beta \sim \mu'}\bigg[\mathcal{T}(\beta) \text{is falsified on } \beta\bigg] \leq \frac{2}{5}.$$
\end{lemma}
\begin{proof}
    Assume $\mathcal{T}$ makes an error with probability $<3/5$. Then our main idea is that we would be able to construct an ordinary decision tree for $\Search(\SP_{n,\rho})$ of depth $ O(n^{1/3})$ which makes error with probability $<7/10$ under distribution $\mu$. This contradicts Lemma~\ref{lemma:DT-error}.
    
    Using Observation~\ref{rem:pdt-to-protocol}, we get a deterministic 2-party protocol of cost at most $2\cdot\text{depth}(\mathcal{T})$ that makes error less than $3/5$ for solving $\Search\big(\SP_{n,\rho}\circ \IP\big)$.
    Theorem~\ref{thm:CFKMP-simulation} then yields a randomized decision tree $\cT'$ with the following properties. On input $\alpha \sim \mu$, 
    \begin{itemize}
        \item $\cT'$ makes at most $O \left(\text{cost}(\Pi)/\log n \right)$ queries to $\alpha$, i.e., at most $O(n^{1/3})$.
        \item If $\mathcal{D}_1$ denotes the actual distribution of the transcript of $\Pi$ when it is run on input sampled uniformly at random from $\overrightarrow{\text{IP}}^{-1}(\alpha)$ and $\mathcal{D}_2$ denotes the distribution of the transcript of $\Pi$ simulated by $\cT'$, 
        $$ || \mathcal{D}_1- \mathcal{D}_2 || \leq \dfrac{1}{10}$$
    \end{itemize}
    We now modify $\cT'$ to output a clause as follows. A transcript of $\Pi$ leads it to output a clause of $\SP_{n,\rho}\circ \IP$. 
    The modified $\cT'$ outputs the unique corresponding un-lifted clause of $\SP_{n,\rho}$. The probability of error is at most $ \text{Pr} [\Pi \text{ errs }] + 1/10 < 7/10$. This gives a randomized decision tree; by fixing the coins we can replace it by a deterministic decision tree, contradicting Lemma~\ref{lemma:DT-error}.
\end{proof}

\subsection{Foolable Nodes Are Frequent}   \label{sec:foolable-nodes-frequent}
Let $\cP$ be a bottom-read-once linear branching program for $\Search\big(\SP_{n,\rho} \circ \IP\big)$ that corresponds to a bottom-regular ResLin proof of the unsatisfiability of $\SP_{n,\rho} \circ \IP$. 
Our goal is to show size of $\cP$ is large. 
To do so, we establish that the affine spaces associated with many nodes of $\mathcal{P}$ have a certain property that allows to fool them. 
In particular, let $\beta$ be an assignment sampled by $\mu'$.
We will prove that with high probability after making $t = O(n^{1/3})$ linear queries, according to $\beta$, we will end in a node $v$ of $\cP$ such that the associated affine space $A_v$ does not have much information about $\beta$.
We will show that this implies the affine space $A_v$ contains many useful assignments that allows us to prove the co-dimension $A_v$ is large.
Now, we define the sought property formally.

\begin{definition}  \label{defn:foolable-assignment}
Let $\alpha \in \Supp(\mu)$ and $P$ be the blue path induced by $\alpha$ that ends at $v$.
Let $u$ and $w$ be the two children of $v$. We say a subset $T \subseteq [m]$ is \emph{$\alpha$-foolable} if $T$ does not contain any variable mentioning $v,u$ or $w$, i.e. the variables $P_{x,i}, Z_{x,j}$ for $x \in \{u,v,w\}, i \in [N], j \in [N-1]$. 
\end{definition}

\begin{definition}  \label{defn:foolable-space}
    An affine space $A \subseteq \F^{mb}_2$ is \emph{$\alpha$-foolable} if $\Cl(A)$ is $\alpha$-foolable with $\alpha \in \Supp(\mu)$ and there exists $\beta \in A$ such that $\alpha = \overrightarrow{\IP}(\beta)$.
\end{definition}

We call a node $v$ of $\cP$ \emph{$\alpha$-foolable} if the associated affine space $A_v$ is $\alpha$-foolable. 
Recall that $\cP(\beta,t)$ is the node that $\cP$ arrives at after making $t$ linear queries on $\beta$. 
It turns out the node $\cP(\beta, t)$ is $\alpha$-foolable with high probability if $t$ is sufficiently small. We prove this in the following theorem.

\begin{theorem}  \label{thm:node-foolability}
     Let $\cP$ be any bottom-read-once linear branching program corresponding to a bottom-regular ResLin proof of $\SP_{n,\rho} \circ \IP$. 
     There exists a constant $c > 0$ such that if $t < c \cdot n^{1/3}$, then 
     $$ \Pr_{ \alpha \sim \mu, \beta \sim \overrightarrow{\IP}^{-1}(\alpha)}\bigg[ \mathcal{P}(\beta,t) \text{ is $\alpha$-foolable}\bigg] > \frac{3}{5}.$$
\end{theorem}
\pavel{I'd call this lemma}
\begin{proof}
    Let $v$ denote the random node $\mathcal{P}(\beta,t)$. Let the blue path induced by $\alpha$ end at $a$ and let the children of $a$ be $b$ and $c$.  Notice that the second condition of being $\alpha$-foolable (that $A_v$ contains an element of $\overrightarrow{\IP}^{-1} (\alpha)$) is always satisfied by $\cP(\beta, t)$ since one such element is $\beta$. We just need to lower bound the probability of the first condition of $\alpha$-foolability being satisfied, i.e., the probability that $\Cl(A_v)$ does not contain any variable mentioning $a,b$ or $c$. \newline
    
    We construct a PDT $\mathcal{T}$ for $\Search(\SP_{n} \circ \IP)$ from $\mathcal{P}$ in the following manner: on input $\beta$, it will simulate the path traced out in $\mathcal{P}$ for $t$ steps by making precisely those linear queries that would have been issued in $\mathcal{P}$. At the end of it, $\mathcal{T}$ does the following: Let $A$ be the affine space corresponding to the queries issued and responses received so far. For every vertex $k= (i,j)$  in the pyramid graph $G_n$ such that one of its variables ($P_{k,\ell}$ or $Z_{k,\ell}$ for some $\ell \in [N]$) is in $\Cl(A)$, query the $b$ coordinates from the blocks of the following set of variables: 
    \[
    S = \left\{R_{f^{-1} (w) } | w \in \{ (i-1, j-1), (i-1, j), (i, j-1), (i, j), (i, j+1), (i+1, j), (i+1, j+1) \} \cap V(G) \right\}
    \]
    If one of the induction clauses mentioning only stones in $S$ is falsified (recall that the placement of stones to vertices is the same for all assignments in $\mu$), output the corresponding clause. Otherwise, output an error symbol. 

    Clearly, the depth of $\mathcal{T}$ is $O(n^{1/3}\log n)$. Thus, by Lemma~\ref{lemma:PDT-average}, the probability that it outputs a falsified clause is at most $2/5$. Let $A_v$ denote the affine space at $\mathcal{P}(\beta, t)$. Note that $A \subseteq A_v \implies \Cl (A_v) \subseteq \Cl (A) $, by Lemma \ref{lem:MonotoneClosure}. It is straight-forward to verify that if $A_v$ is not $\alpha$-foolable, then $\mathcal{T}$ successfully outputs a clause falsified by $\beta$: if one of the variables belonging to $a,b$ or $c$ is in $\Cl(A)$, in the final step the PDT queries the stones placed on $a,b,c$ and detects that an induction clause at $a$ is falsified. The result now follows from Lemma \ref{lemma:PDT-average}.
\end{proof}

\subsection{Foolability Implies Large Rank}  \label{sec:foolability-to-rank}
In this subsection, we prove there is an obfuscation map $\rho: [N]^3 \to \V$ such that for a foolable node $v$ of a bottom-read-once linear branching program $\cP$ computing $\Search(\SP_{n,\rho} \circ \IP)$, the associated affine space $A_v$ must have large co-dimension.

First, we prove an auxiliary lemma.
Let $T \subseteq [m]$ be a subset of the variables of $\SP_{n, \rho} $. We say a stone $j$ is \emph{marked} by $T$ if $T$ contains a variable that mentions the stone $j$, i.e. $R_j, P_{v,j}$ for any vertex $v \in V$, $P_{f(j),k}$ for any stone $k \in [N]$, or $Z_{f(j),\ell}$ for any $\ell \in [N-1]$.
Let $Q(T) \subseteq [N]$ be the set of stones marked by $T$.

\begin{lemma}
\label{lem:Fooling}
    Let $\rho$ be any obfuscation map and $\alpha \in \Supp(\mu)$. Let $v$ be the vertex at which the blue path induced by $\alpha$ ends. Let $T \subseteq [m]$ be an $\alpha$-foolable subset with $|Q(T)| < N/2$. 
    For any $i,j,k \in [N] \setminus Q(T)$, there exists an assignment $\gamma \in \{0,1\}^m$ extending the restriction $\alpha|_T$ which satisfies all clauses of $\SP_{n,\rho} $ except one of the following two: 
    \begin{align*}
    C_1 &:= \neg P_{u, i} \lor \neg R_{i} \lor \neg P_{w, j} \lor \neg R_{j} \lor \neg P_{v,k} \lor R_k \lor \rho(i, j, k), \text{or} \\
    C_2 &:= \neg P_{u, i} \lor \neg R_{i} \lor \neg P_{w, j} \lor \neg R_{j} \lor \neg P_{v,k} \lor R_k \lor \neg \rho(i, j, k),
    \end{align*}
    where $u$ and $w$ are the out-neighbors of $v$.
\end{lemma}
\begin{proof}
    Let $S \subseteq V(G)$ be the set of vertices in the blue path induced by $\alpha$. 
    We assign stone $k$ to vertex $v$ and stones $i,j$ to $u,w$ respectively. To all other vertices we assign arbitrary stones as long as they are consistent with $\alpha_{|T}$. Formally: pick two stones $\ell_1, \ell_2$ in $[N]\setminus (Q(T) \cup f^{-1}(S) \cup \{i,j,k\}) $. Consider the following map $\text{STONE}: V \rightarrow [N]$. 
    
    \begin{align*}
        \text{STONE}(p)= \begin{cases}
            f^{-1}(p) & \text{ if }f^{-1}(p) \in Q(T) \\
            i & \text{ if }p=v \\
            j & \text{ if }p=u \\
            k & \text{ if }p=w \\
            \ell_1 & \text{ if }f^{-1}(p) \not \in Q(T) \cup \{i,j,k\}, f^{-1}(p) \in S  \\
            \ell_2 & \text{ if } f^{-1}(p) \not \in Q(T) \cup \{i,j,k\} \cup S            
        \end{cases}
    \end{align*}
    Define a coloring map $\text{COLOR}: [N] \rightarrow \{ \text{RED}, \text{BLUE} \}$ as follows:

    \begin{align*}
        \text{COLOR}(s)= \begin{cases}
            \text{RED} & \text{ if } s \in Q(T), f(s) \not \in S \\
            \text{BLUE} & \text{ if } s \in Q(T), f(s) \in S \\
            \text{BLUE} & \text { if }s=i \\
            \text{RED} & \text{ if }s=j \\
            \text{RED } & \text{ if }s=k \\
            \text{BLUE} & \text{ if }s=\ell_1 \\
            \text{RED } & \text{ if }s=\ell_2 \\
            \text{BLUE} & \text{ otherwise}
            \end{cases}
    \end{align*}
    We remark the color used in the last case does not matter as these stones are not used in the stone placement given by the map $\text{STONE}$.
    Let $\gamma \in \{0,1\}^{m}$ be the assignment which sets the variables according to this placement and coloring map, i.e., 
    \begin{align*}
        P_{v,j} &= \begin{cases}
            1 & \text{ if }j=\text{STONE}(v) \\
            0 & \text{ otherwise}
        \end{cases} \\
         Z_{v,j} &= \begin{cases}
            0 & \text{ if }j<\text{STONE}(v) \\
            1 & \text{ otherwise}
        \end{cases} \\
        R_j & = \begin{cases}
            1 & \text{ if } \text{COLOR}(j)= \text{RED} \\
            0 & \text{ if } \text{COLOR}(j)= \text{BLUE}
        \end{cases}
    \end{align*}
    Notice that this assignment is consistent with $\alpha$ on $T$ and it falsifies a single induction clause at $v$: stone $i$ is placed at $v$, stones $j,k$ are placed at $u,w$ respectively; stones $j,k$ are red while stone $i$ is blue. 
    We remark that there is no set of clause in $\SP_{n,\rho}$ which forces the placement of stones to vertices to be bijective.
    Thus, the only clause of $\SP_{n,\rho}$ falsified by $\gamma$ is one of the following: 
    \begin{align*}
    C_1 &:= \neg P_{u, i} \lor \neg R_{i} \lor \neg P_{w, j} \lor \neg R_{j} \lor \neg P_{v,k} \lor R_k \lor \rho(i, j, k), \\
    C_2 &:= \neg P_{u, i} \lor \neg R_{i} \lor \neg P_{w, j} \lor \neg R_{j} \lor \neg P_{v,k} \lor R_k \lor \neg \rho(i, j, k),
    \end{align*}    
    
    \pavel{Maybe a picture of stone placement by $\alpha$, $\alpha'$ and $\alpha_1$ would be useful.}
\end{proof}

For our hard formula, we use an appropriate obfuscation map $\rho$ given by the following lemma.
An analogous lemma was proved by Alekhnovich et al.~\cite{AJPU02} with different constants as their formula is over a slightly smaller set of constants, but still quadratic in $N$.
For sake of completeness, we state the proof in Appendix.

\begin{lemma}
\label{lem:RhoFunction}
    For $N$ sufficiently large, there exists a mapping $\rho: [N]^3 \to \V$ such that for every $Q \subseteq [N]$ with $|Q| \leq N/400$ and every $X \in \V$, there exist $i < j < k \in [N] \setminus Q$ such that $\rho(i,j,k) = X$.
\end{lemma}

We remark that the proof of the following lemma is the only place where we are using an obfuscation map with a certain property (given by Lemma~\ref{lem:RhoFunction}) and also that the branching program computing $\Search(\SP_{n,\rho} \circ \IP)$ is bottom-read-once.
\begin{lemma} \label{lemma:large-rank-stone}
    There exists an obfuscation map $\rho: [N]^3 \to \V$ such that the following holds.
    Let $\beta \in \F^{mb}_2$, $t > 0$, and $p = \cP(\beta, t)$ be a node in a bottom-read-once branching program $\cP$ computing $\Search(\SP_{n,\rho} \circ \IP)$.
    If $p$ is $\alpha$-foolable for $\alpha = \overrightarrow{\IP}(\beta)$, then $\codim(A_p) \geq \min\{\frac{N}{800}, t\}$.
\end{lemma}
\begin{proof}
    Fix $\rho: [N]^3 \to \V$ to be a map with the property guaranteed by Lemma~\ref{lem:RhoFunction}.
    Let $A_p = \Space(M, c)$. Suppose $\rank(M) \leq \frac{N}{800}$.     
    Let $S$ be the set of sinks of $\cP$ reachable from $p$.

    \begin{claim} \label{claim:only_one_falsification} For each variable $Y$ of $\SP_{n,\rho} \circ \IP$ there is a sink in $S$ outputting a clause $D$ such that $Y$ or $\neg Y$ is in $D$.
    \end{claim}

    \begin{proof}[Proof of Claim \ref{claim:only_one_falsification}] 
        We shall show that there exists an assignment $\gamma \in A_p$ such that $\gamma$ falsifies only one clause of $\SP_{n, \rho} \circ \IP$ and that clause contains $Y$ or $\neg Y$. 
        Let $Y \in \text{BLOCK}(Z)$. 
        Let $v$ be the endpoint of the blue path induced by $\alpha$, and let its children be $u$ and $w$.
         
        Let $T = \Cl(A_p)$. 
        Note that $|T| < N/800$, thus $Q(T) < N/400$. 
        By Lemma \ref{lem:RhoFunction}, there exist $i<j<k$ in $[N] \setminus Q(T)$ such that $\rho(i,j,k)= Z.$ 
        By assumption, $T$ is $\alpha$-foolable.
        By Lemma \ref{lem:Fooling}, we can extend the restriction $\alpha_{|T}$ to a full assignment $\gamma \in \F^{m}_2$ such that $\gamma$ does not falsify any clause of $\SP_{n,\rho}$ other than one of the following two clauses:
  \begin{align*}
    C_1 &:= \neg P_{u, i} \lor \neg R_{i} \lor \neg P_{w, j} \lor \neg R_{j} \lor \neg P_{v,k} \lor R_k \lor Z\\
    C_2 &:= \neg P_{u, i} \lor \neg R_{i} \lor \neg P_{w, j} \lor \neg R_{j} \lor \neg P_{v,k} \lor R_k \lor \neg Z.
    \end{align*}        
    By Lemma \ref{lem:StiflingGadget}, there is an assignment $\beta' \in A_p$ such that $\overrightarrow{\IP}(\beta')= \gamma$. 
    Note that $\beta'$ falsifies only one clause of $\SP_{n,\rho} \circ \IP$, and that clause belongs to either $C_1 \circ \IP$ or $C_2 \circ \IP$. 
    By Observation~\ref{obs:AllVariablesInLift}, every clause in $C_1 \circ \IP$ and $C_2 \circ \IP$ contains every variable in the block of $\rho(i,j,k)=Z$ (possibly with negations). Thus, the only clause falsified by $\beta'$ contains either $Y$ or $\neg Y$. 
    Since $\beta' \in A_p$ and $\beta'$ falsifies only one clause $\tilde{C}$ of $\SP_{n,\rho} \circ \IP$, one of the sinks in $S$ must be labelled $\tilde{C}$. 
    \end{proof}

    We continue the proof of Lemma \ref{lemma:large-rank-stone}.    
    Let $U$ be the space spanned by rows of matrices defining the spaces of sinks in $S$.
    Formally, let $s \in S$ be labelled by the affine space $A_s = \Space(M_s, c_s)$ (this corresponds to a clause of $\SP_{n,\rho} \circ \IP$).
    Then, $U = \Span\bigl(\{\Row(M_s) \mid s \in S\}\bigr)$.
    By Claim~\ref{claim:only_one_falsification} each variable of $\SP_{n,\rho} \circ \IP$ is mentioned at some sink in $S$, so the space $U$ has full dimension, i.e., $\dim(U) = mb$.
    Let $W = \Span\bigl(\Row(M) \cup \Post(p)\bigr)$.
    By Lemma~\ref{lem:PostDim}, we have
    \[
    \dim(W) \leq \rank(M) + \dim(\Post(p)) \leq \rank(M) + mb - t.
    \]
    On the other hand by Lemma~\ref{lem:BPPath}, $U \subseteq W$ and thus, $\dim(W) \geq \dim(U) = mb$.
    Putting both inequalities together, we get $\rank(M) \geq t$.

 \end{proof}

\subsection{Lifted Distributions Fool Rank}
In the earlier two subsections, we have established the following two facts: (i) in any BROLBP $\mathcal{P}$ corresponding to a bottom-regular ResLin proof of $\SP_{n,\rho} \circ \IP$, when inputs $\beta$ are sampled according to the $\IP$ lift of $\mu$, the node $\mathcal{P}(\beta,t)$ is a foolable node with high probability; (ii) in such a BROLBP, the constraint matrix for the affine space associated with a foolable node has large rank. 

To prove that $\mathcal{P}$ has large size, it is sufficient to argue that each large rank constraint system is satisfied with small probability under the $\IP$ lift of $\mu$.
Of course, such a statement is well known to be true if we sample inputs from the uniform distribution in $\mathbb{F}_2^{bm}$. 
However, our distribution is not so at all. In particular, it has quite sparse support. 
Still, it turns out that any lifted distribution is pseudo-random with respect to the rank measure if the gadget satisfies a generalization of the stifling property. 
Consider a gadget $g:\{0,1\}^b \to \{0,1\}$. For any $i \in [b]$, and $o \in \{0,1\}$ an assignment $\alpha$ to the bits different from $i$ is called \emph{$o$-stifling for $i$} if $g$ gets fixed to $o$ by $\alpha$, i.e., the induced subfunction $g_{|([b] \setminus \{i\}) \leftarrow \alpha}$ gets fixed to the constant function that always evaluates to $o$, no matter how the $i$-th bit is set.  We say $g$ is \emph{$\varepsilon$-balanced, stifled} if for any $i \in [b]$, and for any $o \in \{0,1\}$, the following is true: when we sample $x \in \{0,1\}^b$ uniformly at random from $g^{-1}(o)$, the projection of $x$ on co-ordinates different from $i$ is $o$-stifling for $i$ with probability at least $\varepsilon$. 

\begin{claim}  \label{claim:IP-gadget-quality}
     Inner-product defined on $2b\ge8$ bits is a $7/18$-balanced and stifled gadget.
\end{claim}

\begin{proof}
By simple manipulation,
$$\Pr_{(x,y) \sim \{0,1\}^{2b}} \big[ \IP(x,y) = 0 \wedge x_1 = 0\big] = \mathbb{E}\bigg[ \bigg(\frac{1 + (-1)^{\sum_{i=1}^b x_iy_i}}{2}\bigg)\bigg(\frac{1+(-1)^{x_1}}{2}\bigg)\bigg]. $$ 
The RHS becomes,
$$\frac{1}{4} + \frac{1}{4}\mathbb{E}\big[(-1)^{x_1}\big] + \frac{1}{2}\mathbb{E}\big[(-1)^{\sum_{i=1}^b x_i y_i}\big]$$
The second sum is 0, and the third is at most $\frac{1}{2^{b+1}}$.  Overall, this gives that 
$$\Pr_{(x,y) \sim \{0,1\}^{2b}} \big[ \IP(x,y) = 0 \wedge x_1 = 0\big] \geq \frac{1}{4} - \frac{1}{2^{b+1}}.$$ Further, by a similar method, 
$$\Pr_{(x,y) \sim \{0,1\}^{2b}} \big[ \IP(x,y) = 0\big] \leq \frac{1}{2} + \frac{1}{2^b}.$$
Thus, 
$$\Pr_{(x,y) \sim \{0,1\}^{2b}} \big[ x_1 = 0\,|\, \IP(x,y) = 0 \big] \ge \frac{1/4 - 1/2^{b+1}}{1/2 + 1/2^b}. $$
Note that any assignment that sets $x_1=0$ stifles $y_1$. Hence,  $y_1$ is stifled with probability at least $7/18$,  if $b \geq 4$, by the projection of a random $0$ (and similarly 1) assignment to $\IP$. Completely analogously, any bit is stifled with the same probability.
\end{proof}

\begin{remark}
    It is worth noting that several gadgets, including Inner-Product, Indexing, Majority, even on sufficiently large but constant number of bits, are stifled and balanced.
\end{remark}

Now we state the utility of balanced, stifling gadgets.

\begin{lemma}  \label{lemma:rank-fooling}
    Let $g$ be any $\varepsilon$-balanced, stifled gadget and $z \in \mathbb{F}_2^m$ be any fixed vector. Then, for every matrix $M \in \mathbb{F}_2^{r \times bm}$ of full rank $r$, and sny vector $\gamma \in \mathbb{F}_2^{r}$ the following holds:
    $$\Pr_{\beta \sim \overrightarrow{g}^{-1}(z)} \bigg[M\beta\,=\,\gamma\bigg] = 2^{-\Omega_{\varepsilon}(r/b)}. $$
\end{lemma}
\begin{proof}
    After Gaussian elimination, turning $M$ into row-echelon form, there are at least $r/b$ different blocks in which pivots of rows appear. Let us call each such block a pivot block. The distribution $\overrightarrow{g}^{-1}(z)$ samples independently at random from  $g^{-1}(z_i)$ for each of the $i$-th block. It will be convenient to think that we sample, one after the other, independently from blocks in this way, starting from the rightmost.  Consider the situation when we arrive at a pivot block having sampled all blocks to its right.  Let the bit of $z$ corresponding to this pivot block be $o \in \{0,1\}$.  Consider any one row that has a pivot in that block. Let the equation corresponding to this row be denoted by $\ell$. By the property of $g$, with probability at least $\varepsilon$ the random assignment from $g^{-1}(o)$ will be stifling for the pivot of $\ell$. Conditioned on that event, the stifled bit will be set to each of $0,1$ with probability exactly $1/2$ as each possible setting gives rise to a distinct assignment in $g^{-1}(o)$. Hence, the probability that $\ell$ is satisfied by the sampled assignment to this block is at most $(1-\varepsilon/2)$. Thus, continuing this way, the probability that all the equations are satisfied  is at most $\big(1-\varepsilon/2\big)^{r/b}$, yielding the desired result.
\end{proof}

\subsection{Putting Everything Together}
Now, we are ready to finish the proof of our lower bound, i.e., Theorem~\ref{thm:LowerBound}.
\begin{proof}[Proof of Theorem~\ref{thm:LowerBound}]
    Let $\mathcal{P}$ be the BROLBP derived from a bottom-regular ResLin proof of  $\SP_{n,\rho} \circ \IP$. Let $t = \lfloor c\cdot n^{1/3} \rfloor$ for an appropriately chosen small constant $c > 0$. Combining Theorem~\ref{thm:node-foolability} and Lemma~\ref{lemma:large-rank-stone}, we get 
    \begin{eqnarray}  \label{eq:large-rank-probability}
        & \Pr_{\beta \sim \mu'} \bigg[ A_{\mathcal{P}(\beta,t)} \text{ has co-dimension } \ge t \bigg] & \ge \frac{3}{5}.
    \end{eqnarray}
    On the other hand, for any node $v$ of $\mathcal{P}$ which has $\text{co-dim}(A_v) \ge t$, Lemma~\ref{lemma:rank-fooling} yields,
    \begin{eqnarray} \label{eq:large-rank-prob-upper-bound}
      \Pr_{\beta \sim \mu'} \bigg[ \mathcal{P}(\beta,t) \text{ is } v \bigg] \le 2^{-\Omega\big(\frac{t}{\log n}\big)}.  
    \end{eqnarray}
    If $s$ is the total number of nodes of $\mathcal{P}$, combining \eqref{eq:large-rank-probability} and \eqref{eq:large-rank-prob-upper-bound}, we get immediately
    $$ s \cdot 2^{-\Omega\big(t/\log(n)\big)} \geq \frac{3}{5}.$$
    Substituting the value of $t$ in the above, the result immediately follows.
    
\end{proof}

\section{Future Directions}  \label{sec:conclusion}

We provided the first super-polynomial separation between the powers of bottom-regular and general ResLin proofs. We believe the general proof strategy that we implemented, modifying and generalizing the recent technique of Efremenko, Garl{\'i}k and Itsykson \cite{EGI23}, should yield exponential lower bounds on the length of bottom-regular ResLin proofs for other formulas as well. For instance, formula $\text{MGT}_n$ which is the constant-width version of $\text{GT}_n$, that encodes the contradiction that a finite total order has no minimal element, when obfuscated appropriately and then lifted with inner-product can be proved to be hard for bottom-regular ResLin. Indeed, the first part of our proof strategy is quite general and applies to all ResLin proofs without any assumption on regularity. Here, we need just the fact that the search for a falsified clause in the base formula is hard on average for small height decision trees w.r.t some distribution $\mu$. That is sufficient, thanks to lifting theorems, to yield the fact that after the first few (typically $n^{\Omega(1)}$) linear queries, the branching program corresponding to the ResLin proof is still far from discovering a falsified clause in the lifted formula with high probability, when the input is sampled from the lifted distribution $\overrightarrow{g}^{-1}(\mu)$. This is step 2 of our proof outline. In most formulas, one could then define some natural notion of \emph{foolability} and then say the affine space associated with nodes of the branching program are frequently foolable. How do we know this is useful? Unfortunately, the usefulness of these notions seem formula-specific. For the binary pigeonhole principle, Efremenko et al. observed that local consistency was enough to yield high rank of the dual of the affine space. For our formula, we achieved the same exploiting the obfuscation map and stifling nature of our lifting gadget. But this seems not immediately generalizable. An interesting direction here is the following:
\begin{problem}  \label{problem:tseitin}
Prove strong lower bounds on the size of bottom-regular proofs for the lift of Tseitin formulas over expander graphs.
\end{problem}

Another direction is to consider the formulas where even implementing the first step of our strategy seems impossible. 
\begin{problem}  \label{problem:random}
    Prove strong lower bounds on size of bottom-regular ResLin proofs for appropriate lifts of random constant-width CNF formulas.
\end{problem}

The above seems challenging as for random formulas of constant-width, for every distribution, there exists a decision tree that finds in $O(1)$ queries a falsified clause with high probability. This, very likely, requires changing our technique substantially. Finally, one of the challenges posed by Gryaznov et al. \cite{GPT22} remains still open.
\begin{problem}  \label{problem:top-regular}
   Prove super-polynomial lower bounds on the size of top-regular ResLin proofs.
\end{problem}

\bibliographystyle{plain}
\bibliography{RegularReslin}

\appendix
\section*{Appendix}
\begin{proof}[Proof of Lemma~\ref{lem:RhoFunction}]
    We prove a random mapping has the sought property with non-zero probability.
    We choose $\rho$ uniformly randomly from all possible mappings from $[N]^3$ to $\V$.
    Let $X \in \V$ and $Q \subseteq [N]$ with $|Q| \leq \varepsilon N$ for $\varepsilon = \frac{1}{400}$.
    Recall that the number of variable of $\SP_{n,\rho}$ is $m = 2N^2$.
    Then, we have
     \begin{align*}
     \Pr & [\text{For all } i < j < k \in [N] \setminus Q : \rho(i,j,k) \neq X] = \left(\frac{m -1}{m}\right)^{\binom{(1-\varepsilon)N}{3}} \\    
     & \leq \left(1 - \frac{1}{m}\right)^{(1-\varepsilon)^3N^3 / 12} \leq e^{-(1-\varepsilon)^3N / 24}
     \end{align*}
    By union bound over all possible $X \in \V$ and $Q \subseteq [N]$ , we get the following bound for the probability that there is $X \in V$ and $Q \subseteq [N]$ of size $\varepsilon N$ such that for all triples $i < j < k \in [N] \setminus Q$ holds that $\rho(i,j,k) \neq X$.
    \[
    2N^2 \cdot \binom{N}{\varepsilon N} \cdot e^{-(1-\varepsilon)^3 N / 24} \leq 2N^2 \cdot e^{(H(\varepsilon) - (1-\varepsilon)^3/24)N} \leq 2N^2\cdot e^{-0.0238 N}
    \]
    The first inequality holds because $\binom{N}{\varepsilon N} \leq e^{H(\varepsilon)N}$ for $H(\varepsilon) = - \varepsilon \ln \varepsilon - (1-\varepsilon) \ln (1 - \varepsilon)$.
    The second inequality holds for $\varepsilon = \frac{1}{400}$.
    The last term is clearly less than 1 for large enough $N$.
\end{proof}

\begin{proof}[Proof of Lemma \ref{lemma:resolution-complete}] Suppose the set of clauses $\Phi$ semantically implies $C:=[x_1=a_1] \lor \cdots \lor [x_k=a_k]$. Then, the CNF formula $\Phi' := \Phi \land [x_1 \neq a_1] \land \cdots \land [x_k \neq a_k]$ is unsatisfiable. We construct a decision tree $\cT$ for $\Search(\Phi')$ as follows: 
\begin{description}
    \item[Step 1:] First, $\cT$ queries the variables $x_1, \dots , x_k$. If it sees that $x_j=a_j$ for some $j \in [k]$, it outputs the clause $[x_j \neq a_j]$. 
    \item[Step 2:] Then, $\cT$ queries the variables of $\Phi'$ one-by-one and outputs a falsified clause as soon as it finds one.

\end{description}
By the well-known equivalence between tree-like resolution and decision trees, $\cT$ can be converted into tree-like resolution refutation for $\Phi'$. Each node $v$ of $\cT$ is labelled with a clause $C_v$. If a node queries variable $x_i$ and has two children $v_0, v_1$, then $C_v$ can be derived from $C_{v_0}, C_{v_1}$ by resolving $x_i$. The root is labelled with the empty clause. 

Let $v$ be the node of $\cT$ where Step 2 starts. The clause $C_v$ must be $C= [x_1= a_1] \lor \cdots \lor [x_k=a_k]$, because resolving it with $[x_1 \neq a_1], \dots , [x_k \neq a_k]$ results in the empty clause. Note that every leaf under $v$ is a clause of $\Phi$. 
Moreover, the subtree rooted in $v$ has size at most $2^n$.
Thus, $C=C_v$ can be derived from $\Phi$ in at most $2^n$ steps.
\end{proof}
\end{document}